\documentclass[11pt,a4paper,reqno]{amsart}
\usepackage{amssymb,amsmath,amsthm,amsfonts,amscd}
\usepackage{graphicx,xcolor}
\usepackage{mathrsfs}
\usepackage{hyperref}
\hypersetup{colorlinks=true,allcolors=[rgb]{0,0,0.6}}
\usepackage[utf8]{inputenc}
\usepackage[english]{babel}
\usepackage{dsfont}
\usepackage{enumerate}
\usepackage{setspace}
\usepackage{mathabx}
\usepackage{bm}

\usepackage{geometry}
\renewcommand{\Im}{\mathrm{Im}}
\renewcommand{\Re}{\mathrm{Re}}
          \newtheorem{thm}{Theorem}[section]
          \newtheorem{proposition}[thm]{Proposition}
          \newtheorem{lemma}[thm]{Lemma}
          \newtheorem{corollary}[thm]{Corollary}

          \theoremstyle{definition}

          \newtheorem{remark}[thm]{Remark}

\newcommand{\hfree}{H_{\mathrm{free}}}

\numberwithin{equation}{section}
\newcommand{\bdm}{\begin{displaymath}}
\newcommand{\edm}{\end{displaymath}}
\newcommand{\bdn}{\begin{eqnarray}}
\newcommand{\edn}{\end{eqnarray}}
\newcommand{\bay}{\begin{array}{c}}
\newcommand{\eay}{\end{array}}
\newcommand{\ben}{\begin{enumerate}}
\newcommand{\een}{\end{enumerate}}
\newcommand{\beq}{\begin{equation}}
\newcommand{\eeq}{\end{equation}}
\newcommand{\beqn}{\begin{eqnarray}}
\newcommand{\eeqn}{\end{eqnarray}}
\newcommand{\bml}[1]{\begin{multline} #1 \end{multline}}
\newcommand{\bmln}[1]{\begin{multline*} #1 \end{multline*}}

\renewcommand{\leq}{\leqslant}
\renewcommand{\geq}{\geqslant}

\newcommand{\disp}{\displaystyle}
\newcommand{\tx}{\textstyle}

\newcommand{\supp}{\mathrm{supp}}
\newcommand{\one}{\mathds{1}}
\newcommand{\lf}{\left}
\newcommand{\ri}{\right}
\newcommand{\bra}[1]{\lf\langle #1\ri|}
\newcommand{\ket}[1]{\lf|#1 \ri\rangle}
\newcommand{\braket}[2]{\lf\langle #1|#2 \ri\rangle}
\newcommand{\braketr}[2]{\lf\langle #1\lf|#2\ri. \ri\rangle}
\newcommand{\braketl}[2]{\lf.\lf\langle #1\ri|#2 \ri\rangle}
\newcommand{\mean}[3]{\bra{#1}#2\ket{#3}}
\newcommand{\meanlr}[3]{\lf\langle #1\lf|#2\ri|#3\ri\rangle}

\newcommand{\xv}{\mathbf{x}}
\newcommand{\xvp}{\mathbf{x}^{\prime}}
\newcommand{\yv}{\mathbf{y}}
\newcommand{\kv}{\mathbf{k}}
\newcommand{\kvp}{\mathbf{k}^{\prime}}
\newcommand{\gv}{\mathbf{g}}

\newcommand{\nv}{\mathbf{n}}
\newcommand{\mv}{\mathbf{m}}

\newcommand{\diff}{\mathrm{d}}
\newcommand{\eps}{\varepsilon}

\newcommand{\ceps}{c_{\eps}}
\newcommand{\dist}{\mathrm{dist}}

\newcommand{\hamf}{H_{\mathrm{free}}}
\newcommand{\Z}{\mathbb{Z}}
\newcommand{\R}{\mathbb{R}}
\newcommand{\M}{\mathscr{M}}
\newcommand{\Mfin}{\mathscr{M}_{\mathrm{fin}}}
\newcommand{\N}{\mathbb{N}}

\newcommand{\D}{\mathcal{D}}
\newcommand{\dom}{\mathscr{D}}

\newcommand{\T}{\mathcal{T}}
\newcommand{\HH}{\mathcal{H}}
\newcommand{\HHe}{\mathcal{H}_{\mathrm{eff}}}
\newcommand{\KKe}{\mathcal{K}_{\mathrm{eff}}}

\newcommand{\Q}{\mathcal{Q}}

\newcommand{\QQ}{\mathcal{Q}}

\newcommand{\hh}{\mathfrak{H}}
\newcommand{\hilb}{\mathscr{H}}
\newcommand{\WW}{\mathcal{W}}

\newcommand{\fock}{\Gamma_{\mathrm{sym}}}

\begin{document}
\bibliographystyle{amsalpha}

\title[Effective Potentials in the Quasi-Classical Limit]{Effective Potentials Generated by Field Interaction in the Quasi-Classical Limit}

\author{Michele Correggi}\address{Dipartimento di Matematica ``G. Castelnuovo'', Universit\`{a} degli Studi di Roma ``La Sapienza''; P.le Aldo Moro 5, 00185, Roma, Italy.}\email{michele.correggi@gmail.com}\urladdr{http://www1.mat.uniroma1.it/people/correggi/}
\author{Marco Falconi}\address{Dipartimento di Matematica e Fisica, Università degli Studi Roma Tre; Largo San Leonardo
  Murialdo 1, Palazzo C, 00146, Roma, Italy.}\email{mfalconi@mat.uniroma3.it}\urladdr{http://ricerca.mat.uniroma3.it/users/mfalconi/}

\begin{abstract}
We study the {\it quasi-classical limit} of a quantum system composed of finitely many non-relativistic particles coupled to a quantized field in Nelson-type models. We prove that, as the field becomes classical and the corresponding degrees of freedom are traced out, the effective Hamiltonian of the particles converges in resolvent sense to a self-adjoint Schr\"{o}dinger operator with an additional potential, depending on the state of the field. Moreover, we explicitly derive the expression of such a potential for a large class of field states and show that, for certain special sequences of states, the effective potential is trapping. In addition, we prove convergence of the ground state energy of the full system to a suitable effective variational problem involving the classical state of the field.
\end{abstract}

\subjclass[2010]{}

\date{\today}
 %
\maketitle

\tableofcontents

\section{Introduction}
\label{sec:introduction}

The interaction between particles and radiation, either generated by an electromagnetic field or a phonon
field in a crystal, plays a key role in several phenomena in condensed matter physics \cite{Coh1998}. In
several experiments, however, the presence of a quantum field is even more fundamental, being the core of the
experimental apparatus, {\it e.g.}, acting as a trap to keep the particle confined to a certain region. This
is the typical case of magneto-optical traps, whose relevance goes well beyond low temperature physics
\cite{Ash1997}: such type of confinements of atomic beams \cite{DalCoh2001} has been developed mostly in the
investigation of low temperature behavior of neutral atomic clouds and was involved in one of the first
realizations of Bose-Einstein condensation \cite{Ket1995}. Similar techniques have been used to generate
optical lattices \cite{BDZ2008}, where particles are pinned to lattice sites and can only hop from one site to
another, thus generating a sort of discrete model on the lattice. Concretely this is achieved by superposing
laser beams on a lattice with suitable resonating frequencies.  More recently the same set up has been even
used to generate artificial gauge fields for the atoms \cite{Dal2016}.

The theoretical models conventionally used to describe the atomic systems discussed above (see, {\it e.g.},
\cite{PetSmi2008} and references therein) do not involve, however, the direct interaction between the atoms
or the particles and the quantized radiation field, but rather take the simplified point of view of
approximating such an interaction with effective potentials, {\it i.e.}, of considering directly
Schr\"{o}dinger operators of the form
\beq
	\label{eq:effective so}
	\sum_{j =1}^N \lf( - \Delta_j + V_{\mathrm{eff}}(\xv_j) \ri) + U(\xv_1, \ldots, \xv_N),
\eeq
where $ N $ is the number of quantum particles and $ U $ their interaction potential, {\it e.g.}, Coulomb interaction. The explicit form of the effective potential $ V_{\mathrm{eff}} $ is then tuned appropriately for the specific system under investigation and can range from confining potentials of the form $ |\xv|^s $, $ s \geq 2 $, in the case of magneto-optical traps, to periodic oscillating potentials in the case of optical lattices. For the sake of simplicity we are going to assume that the potential $ U $ satisfies the following assumptions:
\beq
	\label{eq:V}\tag{{\bf A}1}
	U \in L^2_{\mathrm{loc}}\big(\R^{dN};\R^+\big) + K_{\ll}\big(\R^{dN}\big),
\eeq
where $  K_{\ll} $ denotes the set of potentials which are Kato-infinitesimally small w.r.t. the free Laplacian.

The connection between the fundamental Hamiltonian describing quantum particles interacting with a radiation field and the effective model \eqref{eq:effective so} has not attracted much attention, at least in the physics literature, and the justification of \eqref{eq:effective so} is mostly phenomenological. There is however a regime in which such a connection can be put on rigorous grounds and the approximation behind \eqref{eq:effective so} made explicit. This is the semiclassical regime of large number of field excitations (see below), when the quantum nature of field (bosonic) carriers can be neglected and the corresponding degrees of freedom approximated by their classical counterparts. Notice that, in the physical picture we are describing, the quantum nature of the particle system is preserved and only the field is assumed to behave almost classically. We are going to refer to this limit as {\it quasi-classical} limit in order to distinguish it from the usual semiclassical limit. 

The semiclassical approximation of quantum mechanics or Schr\"{o}dinger equation is indeed a widely studied
topic in mathematical physics and we refer to the monographs \cite{Hel1988,Zwo2012} and references therein for
an extensive list of results. On the other hand, the specific case of the quasi-classical limit described
above was studied, to the best of our knowledge, only in \cite{MR2205462}, which focuses on the partially
classical limit of the dynamics in the Nelson model (see below for further comments about this result). From
the technical point of view, the key difference with the conventional results about semiclassics is that the
authors of \cite{MR2205462} have to deal with a classical limit $ \hbar \to 0 $ in an infinite dimensional
Hilbert space (the Fock space of quantized radiation). At that time, only limited mathematical tools were
available to study such a question, whereas more recently semiclassical analysis in infinite dimensions has
been developed for bosonic systems, see, \emph{e.g.}, \cite{ammari:nier:2008,
  MR2513969,Ammari:2014aa,Falconi:2016ab} (also \cite{AJN2015,AN2015}, for a different approach to Weyl
quantization in Wiener spaces). These are in fact the very same techniques we are going to use in this work.

The problem of deriving effective models for the quantum dynamics in a suitable semiclassical limit is clearly
not new in the mathematical physics literature, and we list here some works which have some similarities with
our approach. For instance, in \cite{MR2399613, MR3057191} (see also references therein) the ``opposite''
partial limit of classical particles coupled to a quantized field has been studied. More generally, a regime
in which there emerges a behavior similar to the quasi-classical limit is the adiabatic decoupling generated
by a separation between fast and slow degrees of freedom \cite{PST2003,Teu2003,PST2007}, and also the
non-relativistic limit of electrons coupled to a quantum field \cite{MR1075749,MR1235953,MR1638093}. In spite
of a completely different physical meaning, there are also strong mathematical analogies with the strong
coupling limit for the Fr\"{o}hlich polaron \cite{Frank:2014aa,Frank:2015aa,Griesemer:2016aa} (see also
below). Finally, we want to mention the works \cite{BCFS2007,BCFFS2013} about the effective mass and dynamics
of a quantum particle interacting with the electromagnetic field in QED.


Let us now be more precise about the models we plan to study: we want to focus on the behavior of a quantum system composed by $ N $ non-relativistic particles interacting with a quantized bosonic field, which will be often referred to as radiation. The interaction is modelled by a linear coupling as in the Nelson model \cite{nelson:1190} but we take into account two different cases: either the usual Nelson interaction with ultraviolet cut-off, or the Fr\"{o}hlich polaron model \cite{Frohlich230}.  More precisely, the Hamiltonian of the full system is given by an expression of the following form
\beq
	\label{eq:full ham}
	\framebox{$ H = \hfree + \disp\sum_{j = 1}^N A(\xv_j),	\qquad		\hfree = \HH_0 + \diff \Gamma(\omega), $}
\eeq
where
\beq
	\label{eq:h0}
	\HH_0 =  - \Delta_{\xv_1,\ldots,\xv_N} + U(\xv_1, \ldots, \xv_N),
\eeq
is a self-adjoint operator on $ L^2(\R^{dN}) $, $ d = 1, 2, 3 $. The dispersion relation of the field is $ \omega(\kv) $ and throughout the paper we are going to assume that
\beq
	\label{eq:omega}\tag{{\bf A}2}
	\omega(\kv) \geq 0.		
\eeq
The interaction $ A(\xv) $ is linear in the field creation and annihilation operators\footnote{We denote by $ {\: \cdot \:}^* $ the complex conjugate of a complex number.}, \emph{e.g.},
\beq
	\label{eq:A}
	\framebox{$ A(\xv) = \disp\int_{\mathbb{R}^d} \diff \kv  \lf( a^{\dagger}(\kv) \lambda(\kv) e^{-i \kv \cdot \xv} + a(\kv)  \lambda^{*}(\kv) e^{i \kv \cdot \xv} \ri), $}
\eeq
$ \lambda $ being the Fourier transform of the particle coupling factor, which is assumed to be the same for each particle. Above, $ \diff \Gamma $ stands for the second quantization map and therefore $ \diff \Gamma(\omega) $ is the field energy, \emph{i.e.},
\beq
	\label{eq:field energy}
	\diff \Gamma(\omega) = \int_{\R^d} \diff \kv \: \omega(\kv) \: a^{\dagger}(\kv) a(\kv).
\eeq
 
The physical regime we investigate here is the one which is sometimes referred to as semi-classical limit in
the physics literature \cite{PhysRevB.87.094301,CBO9780511778261A024,KalAubTsi98}, also known as {\it
  quasi-classical} or partially classical limit (see, \emph{e.g.}, \cite{MR2205462,MR1718657}), to distinguish
it from the vast mathematical literature about semiclassics: in the experiments the fields are typically
considered as classical and therefore their quantum nature is discarded. More precisely, we think of a regime
where the number of field excitations, {\it e.g.}, photons or phonons, is large. Hence the non-commutativity
of the quantum variables, which is of order 1 (in units of Planck's constant $ \hbar $), can be neglected when
compared to the large number of excitations. This is the approximation we study here, by proposing a model in
which the classical behavior of the field emerges from the semiclassical limit of a purely quantum system. The
regime is therefore named quasi-classical limit because only the field becomes classical, while the quantum
nature of the particles is preserved.

Such an approximation of large number of field excitations has already been considered in the physics
literature \cite{PhysRevB.87.094301}. This is also the typical case of the strong coupling regime as,
\emph{e.g.}, for the polaron \cite{Frank:2014aa,Frank:2015aa,Griesemer:2016aa}. Alternatively, one can think
of particles whose wave functions live on a scale much smaller than the typical length scale of the field
excitations \cite{KalAubTsi98}.

Concretely, the quasi-classical limit is realized by letting
\beq
	\eps \to 0,
\eeq
where $ \eps $ plays the role of Planck's constant, in the CCR relations satisfied by the annihilation and creation operators $ a(\kv) $ and $ a^{\dagger}(\kv) $, \emph{i.e.},
\beq
	\label{eq:CCR}
	\framebox{$ \lf[ a(\kv), a^{\dagger}(\kvp) \ri] = \eps \delta(\kv - \kvp). $}
\eeq
It is clear that when $\varepsilon\to 0$ the non-commutativity of the field becomes negligible and thus it becomes classical. Notice that such a limit should not be interpreted as a classical limit $ \hbar \to 0 $ but rather as a scale limit emerging from the physics of the coupling. 

Our main goal is thus to identify the effective Hamiltonian of the particles in the limit $ \eps \to 0 $, when
the degrees of freedom of the field are traced out. As we are going to see, we will prove that the system of
particles is still described by a sequence of operators $ \HH_{\eps} $, which converges as $ \eps \to 0 $ in
either the norm or the strong resolvent sense to a self-adjoint Schr\"{o}dinger operator $ \HHe $, given, for
each particle, by the unperturbed particle operator $ \HH_0 $ plus a suitable external potential. Moreover, we provide
the explicit expression of such a potential as a function of the state of the quantized field. 

Once the effective model is identified, it is then natural to ask whether the ground state properties of the
full system can be suitably approximated in terms of the effective operators obtained in the quasi-classical
limit. This is indeed the case for the ground state energy, as we prove for both the massive Nelson model and
the polaron: the effective energy is obtained by minimizing the bottom of the spectrum of the state-dependent
effective Hamiltonian with respect to the classical state of the field.

In this work we are interested in dealing only with the stationary features of the particles and we do not investigate
the full dynamics of the system. Dynamical questions have already been studied under restrictive assumptions
on the initial state: the partial classical limit of time-evolved squeezed coherent states was indeed
considered both for the renormalized Nelson model \cite{MR2205462} and the polaron model
\cite{Frank:2014aa,Frank:2015aa,Griesemer:2016aa}. In the former case, the resulting classical field evolves
freely and the quantum fluctuations are described by a free quantum field together with quantum particles
subjected to an external time-dependent potential given by the classical field. In the latter one, the
quasi-classical limit takes the form of a strong coupling limit and the field does not evolve at order zero
but yields an effective potential on the quantum particles. At first order, for suitable time scales, the
nonlinear Landau-Pekar system is recovered.

In this respect our analysis is more general than the one contained in the works mentioned above
\cite{MR2205462,Frank:2014aa,Frank:2015aa,Griesemer:2016aa}, although we do not address any dynamical
question: in all those papers indeed the initial state of the field must be of very special type, {\it i.e.},
a (squeezed) coherent state, which is already semiclassical from a certain point of view. On the opposite, we
make very weak restrictions on the possible field configurations, and show explicitly how such a freedom
influences the effective Schr\"{o}dinger operator for the particles. Let us also stress that the approximation
of the particle dynamics for generic initial states remains an open problem in both cases and we plan to
address such a question in a future work.

We consider three different forms of interaction, leading to similar results in the partially classic limit, but requiring suitable assumptions and slightly different approaches: 
\ben[1)]
	\item {\bf discrete modes of radiation} (Sect. \ref{sec:discrete}): this is the simplest setting since we assume that the field has only a discrete set of frequencies. It is however meaningful from the physical point of view, since it might be viewed as a model for particles in an optical lattice;
	\item {\bf Nelson model with ultraviolet cut-off} (Sect. \ref{sec:nelson}): the Nelson model \cite{nelson:1190} is simply the continuous version of the previous case, where the high frequencies are cut off by means of a suitable $ \eps$-independent form factor;
	\item {\bf Fr\"{o}hlich polaron} (Sect. \ref{sec:polaron}): this model introduced in \cite{Frohlich230} is typically used to describe the interaction of quantum particles with a phonon field. It might be thought of as a Nelson-type model with a special dispersion relation, where no ultraviolet cut-off is necessary.
\een

We will first focus on the dependence of the effective operator on the chosen state of the full system. A wide class of states, \emph{e.g.}, product states, leads indeed to bounded effective potentials and to norm resolvent convergence of the corresponding Schr\"{o}dinger operators (Sect. \ref{sec:discrete}--\ref{sec:polaron}). For a smaller class of models, {\it i.e.}, the massive Nelson and polaron ones, we will also prove the convergence of the ground state energies (Sect. \ref{sec:ground state}). A rather special choice of the sequence of states, \emph{i.e.}, suitable squeezed coherent states, can generate unbounded potentials, as it occurs for experimental traps (Sect. \ref{sec:traps}).

\bigskip

\noindent
{\bf Acknowledgements.} The authors acknowledge the support of MIUR through the FIR grant 2013 ``Condensed Matter in Mathematical Physics (Cond-Math)'' (code RBFR13WAET). \textsc{M.F.} also thanks \textsc{Z. Ammari} and \textsc{F. Nier} for helpful discussions and comments about semiclassical analysis in infinite dimensional systems.

\section{Main Results}
\label{sec:main}

Before stating our main results, we define more precisely the models we are going to consider. The technical assumptions we make on both the unperturbed part $ \hfree $ of the Hamiltonian of the full system and its interaction terms are also recalled later in Sect. \ref{sec:model}.
 
 As anticipated, we want to consider a coupled system of $N$ quantum $d$-dimensional particles coupled with a bosonic field. Therefore we assume that the Hilbert space is given by 
\beq
	\label{eq:hilbert}
	L^2 (\mathbb{R}^{dN} )\otimes \Gamma_{\mathrm{sym}}\lf(\mathfrak{H}\ri),
\eeq
where $ \Gamma_{\mathrm{sym}}\lf(\mathfrak{H}\ri) $ is the usual bosonic Fock space, \emph{i.e.},
$$
	\Gamma_{\mathrm{sym}}(\mathfrak{H}) = \bigoplus_{n = 0}^{\infty} S_n \lf( \mathfrak{H}\ri)^{\otimes n},
$$
with $ S_n $ the symmetrizing operator. The one-particle space $\mathfrak{H}$ for the field
depends on the model, but it is always a (complex) at most separable Hilbert space. In case of identical particles, $ L^2 (\mathbb{R}^{dN} ) $ can be substituted with either $L^2_{\mathrm{sym}} (\mathbb{R}^{dN} )$ or $L^2_{\mathrm{asym}} (\mathbb{R}^{dN} )$.  For simplicity we take the particle to be spinless, similar arguments may apply to particles with spin and suitable coupling with the field. We will use the following convention throughout the paper: standard capital letters, {\it e.g.}, $  H $, will denote operators on the full Hilbert space $ L^2 (\mathbb{R}^{dN} )\otimes \Gamma_{\mathrm{sym}}\lf(\mathfrak{H}\ri) $, while calligraphic capital letters, {\it e.g.}, $  \HH_0 $,  will stand for operators acting only on the particle Hilbert space $ L^2 (\mathbb{R}^{dN} ) $. Finally we will always use the momentum space representation for the field degrees of freedom and, consistently, all the variables depending on those degrees of freedom will be thought of as functions of $ \kv \in \R^d $, {\it e.g.},
\beq
	a^{\#}(f) : = \int_{\R^d} \diff \kv \: a^{\#}(\kv) f(\kv).
\eeq
We denote by $ \hat{f} $ the Fourier transform of $ f $, {\it i.e.},
\beq
	\hat{f}(\kv) : = \frac{1}{(2\pi)^{d/2}} \int_{\R^d} \diff \xv \: e^{-i \kv \cdot \xv} f(\xv),
\eeq
and by $ \:\check{}\: $ the inverse map.

With the hypotheses described in detail in Sect. \ref{sec:model}, the full Hamiltonian $ H $ given by \eqref{eq:full ham} is self-adjoint on a suitable domain. A form core for $ H $ is any form core domain for the unperturbed part $ \hamf $. We do not discuss  such technical issues further (see again Sect. \ref{sec:model}), in order to state as soon as possible our main results. 
	
Our main goal is to characterize the energy of the particle system once the field degrees of freedom are traced out. Therefore for any product state of the full system of the form 
\beq
	\psi(\xv_1, \ldots, \xv_N) \otimes \Psi_\varepsilon,
\eeq
with $ \Psi_{\eps}  \in \Gamma_{\mathrm{sym}}\lf(\mathfrak{H}\ri)$ normalized, we consider the operator $ \HHe $ acting on $ \psi \in L^2(\R^{dN}) $ defined as the partial trace of the expectation of $ H  $ in the product state above, \emph{i.e.},
\beq
	\label{eq:HHe}
	\framebox{$ \HH_{\eps} : = \mean{\Psi_{\eps}}{H}{\Psi_{\eps}}_{\Gamma_{\mathrm{sym}}\lf(\mathfrak{H}\ri)} - \ceps, $}
\eeq
and study its limit as $ \eps \to 0 $. The constant $ c_{\eps} $ is the mean energy of the field, \emph{i.e.}, explicitly
\begin{equation}
 	\label{eq:ceps}
  	c_{\varepsilon}= \mean{\Psi_{\eps}}{\mathrm{d}\Gamma(\omega)}{\Psi_\varepsilon}_{\Gamma_{\mathrm{sym}}\lf(\mathfrak{H}\ri)},
\end{equation}
and we have subtracted it for simplicity, since it just fixes the zero of the energy scale. Notice that one can take an equivalent point of view and investigate the limit $ \eps \to 0 $ of the quadratic form associated with $ \HHe $ which is defined as
\beq
	\Q_{\eps}[\psi] : = \mean{\psi \otimes \Psi_{\eps}}{H}{\psi \otimes \Psi_{\eps}}_{L^2(\R^{dN}) \otimes \Gamma_{\mathrm{sym}}\lf(\mathfrak{H}\ri)} - \ceps \lf\| \psi \ri\|_{L^2(\R^{dN})}^2.
\eeq

We conclude this preliminary discussion by recalling a result about the convergence of states in the Fock space $ \Gamma_{\mathrm{sym}}\lf(\mathfrak{H}\ri) $ originally proven in \cite{ammari:nier:2008,MR2513969,MR2802894,2011arXiv1111.5918A} (see also \cite{Ammari:2014aa,ammari:15} for further applications). The detailed version of the result is stated in Sect. \ref{sec:semiclassics}. We first identify the subspace of sequences of states always admitting at least one probability measure as a semiclassical accumulation point as the set of states such that the following conditions are satisfied
	\beq
		\label{eq:state}\tag{{\bf A}3}
		\mean{\Psi_\varepsilon}{\mathrm{d}\Gamma(1)}{\Psi_\varepsilon}_{\fock(\hh)} \leq C < + \infty; 	\qquad		\mean{\Psi_\varepsilon}{\mathrm{d}\Gamma(\omega)}{\Psi_\varepsilon}_{\fock(\hh)} \leq C' < +\infty,
	\eeq
{\it i.e.}, the expectation values of both the number operator and $ \diff \Gamma(\omega) $ on such states are uniformly bounded in $ \eps $. Under these assumptions there exists a subsequence $ \lf\{ \Psi_{\eps_k} \ri\}_{k\in \mathbb{N}}$, $\varepsilon_k\to 0$, and a measure $ \mu \in \M(\hh) $, with $ \M(\hh) $ the space of probability measures over $ \hh $, so that, if $ g \in \hh $,
\begin{equation}
  	\label{eq:lim measure}
  	\lim_{k\to \infty} \meanlr{\Psi_{\varepsilon_k}}{a(g) + a^{\dagger}(g)}{\Psi_{\varepsilon_k}}_{\Gamma_{\mathrm{sym}}\lf(\mathfrak{H}\ri)} = 2 \Re \int_{\hh}^{} \mathrm{d}\mu(z) \: \braket{z}{g}_{\hh}.
\end{equation}
In general the limit above is not unique, namely it depends on the chosen subsequence. However we can adopt the following convenient notation: when we write
\beq
	\label{eq:eps convergence}
	\Psi_{\eps} \xrightarrow[\eps \to 0]{} \mu\in \mathscr{M}(\hh),
\eeq
it means that either we are considering a subsequence
$ \lf\{ \Psi_{\varepsilon_k} \ri\}_{k\in \mathbb{N}}$ that converges in the sense of \eqref{eq:lim measure}, or that the function
$\varepsilon\mapsto \Psi_{\varepsilon}$ has a unique limit $\mu$ (no need to extract any subsequence). Since it is always possible to extract at least one convergent
subsequence from the family $ \Psi_{\varepsilon} $, the notation above is justified. 

\subsection{Discrete modes of radiation}
\label{sec:discrete}

The first model consists of a countable number of radiation modes linearly coupled with $ N $ particles (in
$d$ dimensions). In this case the one-particle Hilbert space $ \hh $ is simply $ \ell^2(\Z^d) $. Let
$ \kv_{\nv} \in \R^d $, $ \nv \in \Z^d $, be a collection of real frequencies (modes), characteristic of the
system, and denote $\omega_{\nv}=\lvert \kv_{\nv} \rvert_{}^{}$. Then the full Hamiltonian of system takes the form
\begin{equation}
 	\label{eq:H1}
	H=\disp\sum_{j=1}^N -\Delta_j +U(\xv_1,\ldots,\xv_N)+ \sum_{\nv \in \Z^d} \omega_{\nv} a^{\dagger}_{\nv} a_{\nv} + \sum_{j=1}^N A(\xv_j),
\end{equation}
with
\begin{equation}
 	\label{eq:A1}
 	A(\xv)= \disp\sum_{\nv \in \Z^d} \lf( a^{\dagger}_{\nv} \lambda_{\nv} e^{-i \kv_{\nv} \cdot \xv} + a_{\nv} \lambda_{\nv} e^{i \kv_{\nv} \cdot \xv} \ri),
\end{equation}
where we also assume that
\beq
	\label{eq:g1}\tag{{\bf A}4}
	\lf\{ \lambda_{\nv} \ri\}_{\nv \in \Z^d} \in \ell^2(\Z^d).
\eeq

	\begin{thm}[Effective Hamiltonian]
  		\label{teo:conv1}
  		\mbox{}	\\
		Let the assumptions \eqref{eq:V}, \eqref{eq:omega}, \eqref{eq:state}, and \eqref{eq:g1} be satisfied, and let $ \Psi_{\eps}\to  \mu \in \M(\ell^2(\Z^d)) $ in the sense of \eqref{eq:eps convergence}. Then for any $ \eps $ small, $ \HH_{\eps} $ is a self-adjoint operator on $ \dom(\HH_0) $ and\footnote{We use the shorthand notation $ \lf\| \: \cdot \: \ri\|-\mathrm{res} $ and $ \mathrm{s}-\mathrm{res} $ to indicate the convergence of an operator in norm and strong resolvent sense respectively.}
		\beq
			\label{eq:conv1}
			\framebox{$ \HH_{\eps} \xrightarrow[\eps \to 0]{\lf\| \: \cdot \: \ri\|-\mathrm{res}} \HHe = \HH_0 + \disp\sum_{j = 1}^N V_{\mu}(\xv_j)$,}
		\eeq
		where $ \HHe $ is self-adjoint on $ \dom(\HH_0) $ and
		 \begin{equation}
		 	\label{eq:V1}
 	  	 	\framebox{$ V_{\mu}(\xv) = 2\Re\disp\int_{\ell^2(\mathbb{Z}^d)}^{}\mathrm{d}\mu(z) \: \braketl{\lf\{ \lambda_{\nv} e^{-i \kv_{\nv} \cdot \xv} \ri\}_{\nv \in \Z^d}}{z}_{\ell^2(\Z^d)}. $}
                      \end{equation}
	\end{thm}

Hence the net effect of the field on the particle dynamics in the limit $ \eps \to 0 $ is, in this case, to generate a bounded potential $ V_{\mu} $. The potential depends only on the coupling $\lambda$ between particles and radiation, and on the state of radiation, that in the limiting regime is described by the probability $\mu$. 

Let us now discuss which types of potentials can be obtained in this fashion. The coupling $ \lf\{ \lambda_{\nv} \ri\}_{\nv \in \Z^d} \in \ell^2(\mathbb{Z}^d)$ yields some
\emph{a priori} information on the modes that affect the particles, and with which strength. For instance, if $ \lambda_{\nv}  \neq 0$ for any $ \nv \in \mathbb{Z}^d$, so that every mode contributes to the coupling with
the particles, then $\lambda_{\nv} $ has a multiplicative inverse given by the sequence $ \lambda^{-1} : = \lf\{ \lambda_{\nv}^{-1} \ri\}_{\nv \in \N}$. Let then $ b : = \lf\{ b_{\nv} \ri\}_{\nv \in \N} $ be any sequence in $ \ell^1(\Z^d) $ such that $ \lambda^{-1}_{\nv} b_{\nv} \in \ell^2(\Z^d) $; then we can construct the {\it squeezed coherent states} $ \Xi\lf( \frac{1}{2 \lambda^*} b^{*}\ri) $. For any $ f \in \hh = \ell^2(\Z^d) $, the vector $ \Xi(f)\in \Gamma_s(\hh)$ is given by the usual coherent state for the field, {\it i.e.},
\begin{equation}
	\label{eq:coherent}
  	\Xi(f) := \WW\lf( \tx\frac{1}{i\eps} f\ri)\Omega,
\end{equation}
where $ \Omega $ is the vacuum and $ \WW $ is the Weyl operator
\beq
	\label{eq:weyl}
	\WW(f) : = e^{i \lf( a^{\dagger}(f) + a(f) \ri)}.
\eeq
It is well known \cite[Theorem 4.2]{ammari:nier:2008} that 
\beq
	\label{eq:coherent convergence}
	\Xi(f) \xrightarrow[\eps \to 0]{} \delta\lf(z - f\ri),
\eeq
in the sense defined in \eqref{eq:eps convergence}. Hence, under the above assumptions,
\bdm
	\Xi\lf(\tx\frac{1}{2 \lambda^*} b^{*}\ri) \xrightarrow[\eps \to 0]{} \delta \lf(z - \tx\frac{1}{2 \lambda^{*}} b^*  \ri).
\edm
Therefore, applying Theorem~\ref{teo:conv1} to such a class of coherent states, we obtain the almost periodic
effective potentials
\begin{equation}
  \label{eq:23}
  V_{b}(x)=\sum_{\nv \in \mathbb{Z}^d}^{} \bigl( \Re (b_{\nv}) \: \cos(\kv_{\nv} \cdot \xv) + \Im (b_{\nv})  \sin(\kv_{\nv} \cdot \xv) \bigr).
\end{equation}
These potentials play a very important role in condensed matter experiments, where they take the name of optical lattices \cite{PhysRevA.50.5173,bloch:05}: by suitably tuning superimposed laser beams forming a lattice, one can create periodic wells, which are typically described in first approximation by potentials of the form above. In fact, when the intensity of lasers gets very large, the tunneling between different wells gets small and the particles can be considered pinned at lattice sites, so giving rise to a discrete model. So our result justifies the use of a first quantized periodic potential to approximate the effect of the field interaction in the semiclassical regime. Actually, Theorem \ref{teo:conv1} gives much more information: it is indeed possible to obtain a wider class of almost periodic potentials of the form
\begin{equation}
  	\label{eq:24}
  	2\Re \int_{\ell^2(\mathbb{Z}^d)}^{} \mathrm{d}\mu(z) \: \sum_{\nv \in \mathbb{Z}^d}^{} \lambda^{*}_{\nv} z_{\nv} e^{i \kv_{\nv} \cdot \xv},
\end{equation}
provided that there is a family of quantum states of the field such that $\Psi_\varepsilon \to \mu$. It turns out that all probability
measures $\mu\in \mathcal{\M} \lf( \ell^2(\mathbb{Z}^d) \ri)$ can be reached by suitable families of quantum
states \cite{Falconi:2016ab}.

\subsection{Nelson model: bounded potentials vanishing at infinity}
\label{sec:nelson}

If the radiation has a continuum of modes coupled with particles, it is possible to obtain bounded potentials
vanishing at infinity. 

The one-particle Hilbert space is in this case $ \hh = L^2(\R^d) $ and the Hamiltonian has the form \eqref{eq:full ham} with interaction \eqref{eq:A}, {\it i.e.},
\beq
	\label{eq:nelson ham}
	H=\disp\sum_{j=1}^N -\Delta_j +U(\xv_1,\ldots,\xv_N)+ \int_{\R^d} \diff \kv \: \omega(\kv) a^{\dagger}(\kv) a(\kv) + \sum_{j=1}^N A(\xv_j),
\eeq
\beq
	\label{eq:nelson interaction}
 	A(\xv) = \int_{\R^d}^{} \diff \kv \: \lf( a^{\dagger}(\kv) \lambda(\kv) e^{-i \kv\cdot \xv} + a(\kv) \lambda^*(\kv) e^{i \kv\cdot \xv}\ri);
\eeq
where the cut off is chosen so that
\beq	
	\label{eq:g2} \tag{{\bf A}4${}^{\prime}$}
	\lambda(\kv) \in L^2(\R^d).
\eeq

The analogue of Theorem \ref{teo:conv1} is the following.

	\begin{thm}[Effective Hamiltonian]
  		\label{teo:conv2}
  		\mbox{}	\\
  		Let the assumptions \eqref{eq:V}, \eqref{eq:omega}, \eqref{eq:state} and \eqref{eq:g2} be satisfied and let $ \Psi_{\eps}\to  \mu \in \M(L^2(\R^d)) $ in the sense of \eqref{eq:eps convergence}. Then for any $ \eps $ small, $ \HH_{\eps} $ is a self-adjoint operator on $ \dom(\HH_0) $ and 
		\beq
			\label{eq:conv2}
			\framebox{$ \HH_{\eps} \xrightarrow[\eps \to 0]{\lf\| \: \cdot \: \ri\|-\mathrm{res}} \HHe = \HH_0 + \disp\sum_{j = 1}^N V_{\mu}(\xv_j)$,}
		\eeq
		where $ \HHe $ is self-adjoint on  $ \dom(\HH_0) $ and
		 \begin{equation}
		 	\label{eq:V2}
 	  	 	\framebox{$ V_{\mu}(\xv) = 2 (2\pi)^{d/2} \Re\disp\int_{L^2(\mathbb{R}^d)}^{}\mathrm{d}\mu(z) \: \big( \widecheck{z \lambda^{*}} \big)(\xv). $}
            \end{equation}
  	\end{thm}

Note that, under the assumptions we made, $ z \lambda^{*} \in L^1(\R^d) $ and therefore its Fourier (anti-)transform is a well defined $ L^{\infty} $ function. The allowed effective potentials $ V_{\mu}$ for this model are thus averages of Fourier transforms of $L^1 (\R^d )$ functions and, as such, they are continuous and vanishing at infinity. More precisely, suppose that $\lambda\in L^2 (\R^d )$ has a
multiplicative inverse almost everywhere $\frac{1}{\lambda}$ ({\it e.g.}, $\lambda$ is
not compactly supported), then for any $ \hat{f} \in L^1 (\R^d )$
such that $\frac{1}{\lambda} \hat{f} \in L^2 (\R^d )$, the potentials
\begin{equation}
  	\label{eq:28}
 	V_f(\xv)= \Re  f(\xv)
\end{equation}
are recovered by taking squeezed coherent states of the form 
\bdm
	\Xi\lf(\tx\frac{1}{2(2\pi)^{d/2} \lambda^{*}} \hat{f} \ri).
\edm
Being the Fourier anti-transform of functions in $ L^1 $, such potentials are actually continuous functions vanishing at $ \infty $.

\subsection{Polaron model: form-bounded potentials}
\label{sec:polaron}

Finally, we focus our attention to the Fr\"{o}hlich polaron model \cite{Frohlich230}, which is meant to describe the coupling between electrons and vibration modes in a crystal. The polaron Hamiltonian is
``more singular'' than the other Nelson-type operators previously considered, but the corresponding unitary
dynamics can still be defined without a renormalization procedure. In this model the charge distribution is concentrated at a single point. In the Fock representation the Hilbert space of the theory is, as in
Section~\ref{sec:nelson},
$\mathscr{H}=L^2 (\R^{dN} )\otimes \fock\lf(L^2 (\R^d )\ri)$, with $d\geq 2$. The
Hamiltonian $H$ takes the form
\begin{equation}
  \label{eq:polaron ham}
  H=\disp\sum_{j=1}^N -\Delta_j +U(\xv_1,\ldots,\xv_N)+ \int_{\R^d} \diff \kv \: a^{\dagger}{(\kv)} a{(\kv)} + \sum_{j=1}^N A(\xv_j),
\eeq
\beq
	\label{eq:polaron interaction}
 	A(\xv) = \int_{\R^d}^{} \diff \kv \: \frac{1}{|\kv|^{\frac{d-1}{2}}} \lf( a^{\dagger}(\kv) e^{-i \kv\cdot \xv} + a(\kv) e^{i \kv\cdot \xv}\ri).
\eeq

As in Sections~\ref{sec:discrete} and~\ref{sec:nelson}, for suitably regular states it is possible
to prove the convergence of the effective potential when
$\varepsilon\to 0$.

	\begin{thm}[Effective Hamiltonian]
  		\label{teo:conv3}
  		\mbox{}	\\
  		Let the assumptions \eqref{eq:V} 
  		and \eqref{eq:state} be satisfied and let $ \Psi_{\eps}\to  \mu \in \M(L^2(\R^d)) $ in the sense of \eqref{eq:eps convergence}. Then for any $ \eps $ small, $ \HH_{\eps} $ is a self-adjoint operator on $ \dom(\HH_{\eps}) $ with form domain $\dom(\sqrt{\HH_0})$, and 
  		\beq
			\label{eq:conv3}
			\framebox{$ \HH_{\eps} \xrightarrow[\eps \to 0]{\lf\| \: \cdot \: \ri\|-\mathrm{res}} \HHe = \HH_0 + \disp\sum_{j = 1}^N V_{\mu}(\xv_j),$}
		\eeq
		where $ \HHe $ is self-adjoint on $ \dom(\HHe) $ with form domain $\dom(\sqrt{\HH_0})$, $ V_{\mu} $ is infinitesimally form-bounded w.r.t $ - \Delta $, and
		\begin{equation}
		 	\label{eq:V3}
 	  	 	\framebox{$ V_{\mu}(\xv) = 2 (2\pi)^{\frac{d}{2}} \Re\disp\int_{L^2(\mathbb{R}^d)}^{}  \mathrm{d}\mu(z) \:  \widecheck{ \lf(|\kv|^{\frac{1-d}{2}} z \ri)} (\xv). $}
  		\end{equation} 
	\end{thm}

As before the notation in \eqref{eq:V3} stands for
\bdm
	\widecheck{ \lf(|\kv|^{\frac{1-d}{2}} z\ri)}(\xv) : = \frac{1}{(2\pi)^{\frac{d}{2}}} \int_{\R^d} \diff \kv \: e^{i\kv \cdot \xv}  |\kv|^{\frac{1-d}{2}} z(\kv).
\edm
Note, however, than, unlike the potentials obtained in the case of the Nelson model, $ V_{\mu} $ is in general unbounded and it could not vanish at infinity. Anyways, as stated in the Theorem, $ V_{\mu} $ is infinitesimally form-bounded w.r.t. $ - \Delta $ and therefore it is only an arbitrarily small perturbation of the kinetic energy.

As for the Nelson model, it is interesting to find out which type of potentials can be produced through this quasi-classical limit. By taking suitable squeezed coherent states one can indeed get in the limit $ \eps \to 0 $ a wide class of potentials $ W $. Such potentials $ W $ might not vanish at infinity but can not be trapping in the usual sense, {\it i.e.}, the resolvent of $ -\Delta + W $ can not be compact.  More precisely, let
\bdm
	W \in \dot{H}^{\frac{d-1}{2}}(\R^d) \cap L^2_{\mathrm{loc}} (\R^d),
\edm 
then the squeezed coherent state
\bdm
	\Xi\lf(\tx\frac{1}{2(2\pi)^{d/2}} \lf| \kv \ri|^{\frac{d-1}{2}} \widecheck{W} \ri)
\edm
yields, according to \eqref{eq:V3}, the potential $ W $. In fact in this case the effective potential does not depend on $ \eps $ and equals $ W $ even before the limit $ \eps \to 0 $ is taken. The regularity request on $ W $ is made in order to ensure that the argument of the coherent state is an $ L^2 $ function and therefore the construction makes sense. Note that such potentials are actually the ``static'' analogues of the $0$th-order strongly coupled polaron dynamics studied in~\cite{Frank:2014aa}. 

More in general, all the potentials generated by the field interaction are form-bounded w.r.t. the free part of the Hamiltonian. In fact, as described in Remark \ref{rem:V decay}, if $ \Psi_{\eps} $ is more regular, {\it e.g.}, in addition to \eqref{eq:state} it belongs uniformly to $ \dom (\sqrt{\mathrm{d}\Gamma(|\kv|^2)}) $, then the potential $ V_{\mu} $ is continuous and vanishes as $ |\xv| \to \infty $. Hence we can say that in order to obtain ``rougher'' potentials, the state of the field can not be too regular. Notice however that the effective potential is in any case form-bounded and therefore can never be too strong.

\subsection{Ground state energy}
\label{sec:ground state}

This Section is devoted to the study of the ground state energy of the full Hamiltonian $ H $ in the quasi-classical limit $ \eps \to 0 $. In order to stress the dependence on $ \eps $, in this Section only we set $ H_{\eps} : = H $. All the three types of models considered so far take into account operators $ H_{\eps} $ which are bounded from below. However, in order to state our result, we have to select either the massive Nelson model or the Fr\"{o}hlich polaron (see Remark \ref{rem:massless} below for a discussion of the reasons). 
 
For any self-adjoint operator $A$ on $ \hilb $, we denote by
$\underline{\sigma}(A) \in \mathbb{R} \cup \{-\infty\}$ the bottom of the spectrum of $A$:
\begin{equation}
	\label{eq:bottom}
  	\underline{\sigma}(A) := \inf \lf\{\lambda\in \mathbb{R} \: \big| \: \lambda\in \sigma(A) \ri\} = \inf_{\psi \in \dom(A), \lf\| \psi \ri\|_2 = 1} \mean{\psi}{A}{\psi}_{\hilb},
\end{equation}
where $ \dom(A) \subset \hilb $ is the self-adjointness domain of $ A $ or any core for it.

Our main result is the convergence of the bottom of the spectrum of $ H_{\eps} $ as $\eps \to 0 $ to the infimum of the ground state energy of $ \HHe $ w.r.t. the measure $ \mu $ identifying the classical limit of the state of the field. To this purpose, let us define the measure minimization domain
\beq
	\label{eq:domain measure}
	\M_{\omega} : =\lf\{\mu\in \M \lf(L^2 (\mathbb{R}^d  )\ri) \: \Big| \: \mu\lf(L^2_{\omega} (\mathbb{R}^d )\ri) = 1, \: \mu\big|_{L^2_{\omega} (\mathbb{R}^d )}\text{ is Borel}, \: c(\mu) < \infty \ri\}.
\eeq
Here $ c(\mu) $ is the classical energy of the field $ c(\mu) =\lim_{\eps \to 0} \ceps $, for any $ \Psi_{\eps} \to \mu $, {\it i.e.}, 
\beq
	\label{eq:cmu}
	c(\mu) := \int_{L^2(\R^d)} \diff \mu(z) \: \big\| \omega^{1/2} z  \big\|_2^2 = \int_{L^2(\R^d)} \diff \mu(z) \int_{\R^d} \diff \kv \: \omega(\kv) \lf|z(\kv) \ri|^2.
\eeq
In addition, we have set
\beq
	\label{eq:l2omega}
	L^2_{\omega}(\R^d) : = \bigg\{ f \in L^2(\R^d) \: \bigg| \:  \int_{\R^d} \diff \kv \: \omega(\kv) \lf|f(\kv) \ri|^2 < \infty  \bigg\}.
\eeq
To simplify the presentation we formulate the results only for systems with continuously many radiation modes, for a countable number of modes it can be easily adapted. Finally, recall the definitions of $ \HHe $ given in \eqref{eq:conv2} and \eqref{eq:conv3}, and its dependence on the classical measure $ \mu $ through the potentials \eqref{eq:V2} and \eqref{eq:V3}.

	\begin{thm}[Ground state energy]
		\label{teo:gse}
		\mbox{}	\\
		Let the operator $ H_{\eps} $ be given either by \eqref{eq:nelson ham} or \eqref{eq:polaron ham}. Also, let the assumptions \eqref{eq:V}, \eqref{eq:omega} and \eqref{eq:g2} 
		be satisfied, with the additional request
		\beq
			\label{eq:omega2}\tag{{\bf A}2${}^{\prime}$}
			\omega(\kv) \geq c > 0,	\quad	\mbox{uniformly w.r.t. } \kv \in \R^d, 
		\eeq
		Then we have
		\beq
			\label{eq:gse}
			\framebox{$ \disp\lim_{\eps \to 0} \underline{\sigma}(H_{\eps}) = \inf_{\mu \in \M_{\omega}} \big[ \underline{\sigma} (\HHe) + c(\mu) \big] $.}  
		\eeq
	\end{thm}
	
	\begin{remark}[Boundedness from below]
		\label{rem:boundedness}
		\mbox{}	\\
		Since $ \sigma(H_{\eps}) $ is bounded from below (see Propositions \ref{pro:sa1} and
                \ref{pro:sa2}), and the bound can be actually chosen uniformly w.r.t.  $\varepsilon$, we
                implicitly state that the r.h.s. is also finite (Propositions \ref{pro:bd kke} and
                \ref{pro:bd tmu}). In order for this to be true, the presence of the constant $ c(\mu) $ is
                obviously crucial: the free energy of the field is needed in order to control from below the
                interaction term.
	\end{remark}
	
	\begin{remark}[Nelson massless model]
		\label{rem:massless}
		\mbox{}	\\
		The reason why the Nelson massless model is excluded from the statement is that, without the
                bound \eqref{eq:omega2}, there can be quantum states for which the associated measure $ \mu $
                is concentrated on a suitable homogeneous Sobolev space and hence outside of $L^2$. In fact,
                $ \mu $ is not in general a true probability measure on $L^2$ but only a cylindrical measure
                \cite{Falconi:2016ab}, with respect to which it is possible to integrate only cylindrical
                functions. We still expect the result to be true for the massless Nelson model; the proof
                however would require to deal with such technical problems and we omit its discussion here,
                for the sake of simplicity.
	\end{remark}
	
	\begin{remark}[Convergence of ground states]
		\label{rem:gs conv}
		\mbox{}	\\
		The reader might wonder whether it is possible to deduce from the ground state energy
                convergence \eqref{eq:gse} an analogous result for the ground states. The major obstruction in
                this direction is given by the existence of the ground state itself: it is indeed known that,
                for instance, the massive Nelson model admits a ground state, once the translational symmetry
                has been broken. However, it is much more complicated to show that the infimum of the
                r.h.s. of \eqref{eq:gse} is actually reached on a configuration
                $ \mu_{\mathrm{gs}}, \psi_{\mathrm{gs}} $: for any given measure $ \mu $, the Schr\"{o}dinger
                operator $ \HHe $ certainly has a ground state $ \psi_{\mu} $ but it is far from obvious that it
                would converge on a minimizing sequence $ \mu_n $.
	\end{remark}

\subsection{Trapping potentials}
\label{sec:traps}

We conclude the Section by presenting a generalization of the results discussed in Sect. \ref{sec:discrete}--\ref{sec:polaron}, {\it i.e.}, the convergence of the effective particle Hamiltonians to Schr\"{o}dinger operators with trapping. Indeed, as we have commented extensively, the effective potentials $ V_{\mu} $ obtained in the quasi-classical limit in Theorems \ref{teo:conv1}, \ref{teo:conv2} and \ref{teo:conv3} are never traps. In fact, with the exception of the polaron, those potentials always vanish at infinity.  So in this discussion we take a different point of view: instead of considering a rather general state for the full system, but with good properties in terms of the classical limit, we restrict the class of field configurations to coherent states and drop the regularity assumptions, in order to find out whether one can reproduce a wider class of effective potentials. As we are going to see, this is indeed the case and we will show that one can derive any reasonable confining trap.

Let us now consider the Nelson model defined by \eqref{eq:nelson ham} and recall the definition \eqref{eq:coherent} of a squeezed coherent state:
\beq
	\label{eq:coherent state}
	\Xi(f) := \WW\lf( \tx\frac{1}{i\eps} f\ri)\Omega,
\eeq
where $\Omega$ is the vacuum in $ \fock(L^2 (\mathbb{R}^d ))$ and $\WW(f)$, $f\in L^2 (\mathbb{R}^d )$, the Weyl operator. We have seen in \eqref{eq:coherent convergence} that
\bdm
	\Xi(f) \xrightarrow[\eps \to 0]{} \delta(z - f),
\edm
in the sense of \eqref{eq:eps convergence}. If $ f $ is independent of $ \eps $ and belongs to $ L^2(\R^d) $, the potential generated in the limit is always vanishing at infinity. Therefore we modify
the coherent vector, in such a way that it converges to a point measure on $\D^{\prime}(\mathbb{R}^d)$
concentrated \emph{outside} of $L^2 (\mathbb{R}^d )$. 

We are now ready to state the main result of this section. Let
\beq
	\label{eq:trap}
	W \in L^2_{\mathrm{loc}} (\mathbb{R}^d; \mathbb{R}^+ )
\eeq
be any positive confining potential and assume that $\lambda$ admits a polynomially bounded
multiplicative inverse $\frac{1}{\lambda}$, then we denote by $f_{W,\varepsilon} \in C^{\infty}_0(\R^d) $ the function
\beq
	\label{eq:fw}
 	f_{W,\varepsilon}(\kv)= \frac{1}{2(2\pi)^{d/2} \lambda^{*}(\kv)} \lf(  \widehat{\varphi_{\varepsilon} *  W} \ri)(\kv),
\eeq
where $ \varphi_{\varepsilon}(\xv)=\varepsilon^{-d} \varphi(\xv/\varepsilon) $, $ \varphi \in C^{\infty}_0(\R^d) $, is a suitable mollifier (see Lemma \ref{lemma:l2 convergence} for further details). The coherent state we want to consider has then the form
\beq
	\label{eq:coherent trap}
	\Xi\lf(f_{W,\eps}\ri),
\eeq
and notably it does not satisfy the assumptions \eqref{eq:state}. As a matter of fact, by Proposition~\ref{eq:coherent potential}, it follows that
\begin{align*}
  &\meanlr{\Xi\lf(f_{W,\eps}\ri)}{\mathrm{d}\Gamma(1)}{\Xi\lf(f_{W,\eps}\ri)} = \lf\| f_{W,\eps}  \ri\|_{L^2}^2\; ,\\
  &\meanlr{\Xi\lf(f_{W,\eps}\ri)}{\mathrm{d}\Gamma(\omega)}{\Xi\lf(f_{W,\eps}\ri)} = \lf\| \sqrt{\omega} f_{W,\eps}  \ri\|_{L^2}^{2},
\end{align*}
and both right hand sides diverge as $\varepsilon\to 0$ whenever $W\notin L^2(\R^d)$.

	\begin{thm}[Effective Hamiltonian]
		\label{teo:trap}
		\mbox{}	\\
		Let the assumptions \eqref{eq:V}, \eqref{eq:omega} and \eqref{eq:g1} be satisfied and additionally assume that $ \frac{1}{\lambda(\kv)} $ is polynomially bounded and $ \lambda $, $ \sqrt{\omega} \lambda \in L^2(\R^d) $. Then we have
		\beq
			\label{eq:conv4}
			\framebox{$ \disp\meanlr{\Xi\lf(f_{W,\eps}\ri)}{H}{\Xi\lf(f_{W,\eps}\ri)}_{\fock(L^2(\R^d))} \xrightarrow[\eps \to 0]{\mathrm{s}-\mathrm{res}} \HHe = \HH_0 + \disp\sum_{j = 1}^N W(\xv_j)$,}
		\eeq
    		and $ \HHe $ is essentially self-adjoint on $ C^{\infty}_0(\R^d) $.
	\end{thm}
	
The paradigmatic case one can think of is the derivation of an harmonic trapping potential: $ W(\xv) = \alpha |\xv|^2 $ satisfies indeed the hypothesis of the Theorem and therefore the partial trace of  $ H $ on the coherent state $ \Xi(f_{W,\eps}) $ converges in strong resolvent sense to the Schr\"{o}dinger operator
\bdm
   \HHe =\sum_{j=1}^n \lf( -\Delta_j+\alpha \lf| \xv_j \ri|^2 \ri) + U(\xv_1, \ldots, \xv_N).
\edm
A similar statement holds true for $ W(\xv) = \alpha |\xv|^s $, $ s > 0 $, or, more in general, for any positive potential diverging at infinity. The magneto-optical traps considered in condensed matter physics are then reproduced as effective potentials emerging from the interaction of quantum particles with a radiation field in the quasi-classical regime.

	\begin{remark}[Field energy]
		\label{rem:field energy}
		\mbox{}	\\
		It is interesting to remark that all the confining potentials described above can be obtained
                in the quasi-classical limit only with an infinite energy of the field. More precisely,
                whenever $ W $ is trapping, the free energy $ c_{\eps} \to + \infty $, as $ \eps \to 0 $ :
                recall that for a squeezed coherent state $\Xi(f_{\varepsilon})$, $c_{\varepsilon}$ takes the form
                \begin{equation*}
                  c_{\varepsilon}=\lVert \sqrt{\omega} f_{\varepsilon}  \rVert_2^2
                \end{equation*}
                (see Proposition \ref{pro:coherent states} for further details) and therefore it diverges
                in the limit $\varepsilon\to 0$, whenever $\lim_{\varepsilon\to 0}f_{\varepsilon}\notin L^2$ in
                the distributional sense. This is however not surprising since the physical approximation we
                are considering is the one of large number of field excitations: in order to have a trapping
                effective potential, the field must be very strong. Therefore the number of excitations has to
                diverge even faster and the field energy has to become the dominant term in the energy.
	\end{remark}

\section{Proofs}
\label{sec:proofs}

\subsection{Preliminaries}
\label{sec:model}

We first discuss the well-posedness of the models we plan to study and state the explicit technical assumptions we make. 

The potential $ U $, which is supposed to describe both an additional external trapping and the particle interaction, is assumed to be such that $ \HH_0 = - \Delta + U $ is self-adjoint and bounded from below on $ L^2 (\mathbb{R}^{dN} )$. For concreteness we require
\beq
	\tag{{\bf A}1}
	U \in L^2_{\mathrm{loc}}\big(\R^{dN};\R^+\big) + K_{\ll}\big(\R^{dN}\big),
\eeq
where 
$$ 
	K_{\ll} \big(\R^{dN}\big) = \lf\{ V : \R^{dN} \to \R \: \big| \: V \mbox{ is infinitesimally bounded w.r.t. }  - \Delta \ri\},
$$
is the set of multiplication operators which are Kato-infinitesimally small w.r.t. $ - \Delta $. In the following we will use the notation $ U =: U_+ + U_{\ll}  $ to distinguish the positive part $ U_+ $ of the potential from the infinitesimal one $ U_{\ll} $. With such assumptions $ \HH_0 $ is essentially self-adjoint on $C_0^{\infty}\lf(\R^{dN}\ri)$ and self-adjoint and bounded from below on
$$
	 \dom(\HH_0) = \lf\{ \psi \in H^2(\R^{dN}) \: \big| \: U_+ \psi \in L^2\big(\R^{dN}\big) \ri\}.
$$
We aim at modelling a Coulomb-type interaction between the particles and, possibly, the presence of an external trapping potential, that is assumed to be positive without loss of generality.

Concerning the field part of the free Hamiltonian $ \hfree $, we recall that
\beq
	\tag{{\bf A}2}
	\omega(\kv) \geq 0,		
\eeq
so that $ \diff \Gamma(\omega) $ is a self-adjoint operator on $ \Gamma_{\mathrm{sym}}\lf(\mathfrak{H}\ri) $ with domain $ \dom(\diff \Gamma(\omega)) $.

It remains then to give a meaning to the interaction term. For any $ g(\xv) \in L^{\infty}(\R^{d}; \hh)$, one can easily define the creation and annihilation operators $ a(g(\xv)) $, $ a^{\dagger}(g(\xv )) $ and their sum $ a(g(\xv)) + a^{\dagger}(g(\xv)) $, as closed and densely defined operators on the Fock space  $ \Gamma_{\mathrm{sym}}\lf(\mathfrak{H}\ri) $ for a.e. $ \xv \in \R^d$. 

The simplest case we are going to consider is $ \hh = \ell^2(\Z^d) $, in which case
\beq
	a^{\#}(g(\xv)) : = \sum_{\nv \in \Z^d} a^{\#}_{\nv} g_{\nv}(\xv),
\eeq
with $ a^{\#}_{\nv} $ the usual creation and annihilation operators associated with the frequencies $ \kv_{\nv} \in \R^d $, such that 
\bdm
	\lf[ a_{\nv}, a^{\dagger}_{\mv} \ri] = \varepsilon \delta_{n_1,m_1} \dotsm \delta_{n_d,m_d},
\edm
and 
\beq
	\lf\{ g_{\nv}(\xv) \ri\}_{\nv \in \Z^d} \in \ell^2(\Z^d),	\quad	\mbox{for a.e. } \xv \in \R^d.
\eeq
The dispersion relation is in this case set equal to
\beq
	\omega(\kv_{\nv}) = \omega_{\nv} : = \lf| \kv_{\nv} \ri|.
\eeq

Similarly, when $ \hh = L^2(\R^d) $ (Nelson model), 
\beq
	a^{\#}(g(\xv)) : = \int_{\R^d} \diff \kv \: a^{\#}(\kv) g(\kv;\xv),
\eeq
with $ a^{\#}(\kv) $ the usual operator-valued distributions satisfying \eqref{eq:CCR} and 
\beq	
	g(\: \cdot \:; \xv) \in L^2(\R^d),	\quad	\mbox{for a.e. } \xv \in \R^d.
\eeq

In both cases described above we define the interaction as 
\beq
	\label{eq:interaction}
	\sum_{j = 1}^N A(\xv_j) : = \sum_{j = 1}^{N} \lf[ a(g(\xv_j)) + a^{\dagger}(g(\xv_j)) \ri],
\eeq
with
\beq
	\tag{{\bf A}4}
	g_{\nv}(\xv) = \lambda_{\nv} e^{- i \kv_{\nv} \cdot \xv},		\qquad		\lf\{ \lambda_{\nv} \ri\}_{\nv \in \Z^d} \in \ell^2(\Z^d), 
\eeq
in the first case and 
\beq	
	\tag{{\bf A}4${}^{\prime}$}
	g(\kv; \xv) = \lambda(\kv) e^{- i \kv \cdot \xv},	\qquad		\lambda \in L^2(\R^d),
\eeq
in the second one. The polaron is obviously not covered by the assumptions above and we will discuss it separately.

A preliminary but crucial result for our analysis is the self-adjointness of the full Hamiltonian of the system, which in the case of the Nelson model (and a fortiori for a discrete set of frequencies) can be proven directly using the properties of at most quadratic interactions in the Fock space. We refer to \cite{MR0266533,Falconi:2014aa} for a detailed proof.  We denote by $ C_0^{\infty}(\R^d) $ the set of smooth functions with compact support and, consequently, $ C_0^{\infty}(\diff \Gamma(1)) \subset \fock(\hh) $ stands for the vectors in $ \fock(\hh) $ with finitely many particles.

	\begin{proposition}[Self-adjointness of $ H $ -- cases 1. \& 2.]
  		\label{pro:sa1}
  		\mbox{}	\\
 		Let $ H $ be given by \eqref{eq:full ham} with interaction \eqref{eq:interaction} and let the
                assumptions \eqref{eq:V}, \eqref{eq:omega} and \eqref{eq:g1} (resp. \eqref{eq:g2}) be
                satisfied. Then the Hamiltonian $H$ is essentially self-adjoint on the domain
                $ \dom(\HH_0) \cap \dom\lf(\mathrm{d}\Gamma(\omega) \ri)\cap C_0^{\infty}(\mathrm{d}\Gamma(1)) $. If in
                addition $\omega^{-1/2}\lambda\in \ell^2(\mathbb{Z}^d)$ (resp. $L^2 (\mathbb{R}^d )$), then $H$ is self adjoint
                on $\dom(\HH_0) \cap \dom\lf(\mathrm{d}\Gamma(\omega \ri))$ and bounded from below.
	\end{proposition}
	
	\begin{proof}
          The first part of the statement is a straightforward application of \cite[Theorem
          3.1]{Falconi:2014aa}. Under the additional regularity assumptions on $\lambda$, the exact domain of
          self-adjointness and boundedness from below are obtained via an application of Kato-Rellich Theorem:
          one can indeed show that both $ U_{\ll} $ (by assumption) and the interaction term are infinitesimally
          small w.r.t. to $ \hamf + U_+ $ in the sense of Kato. We postpone the details to the Appendix.
	\end{proof}
	
As anticipated, the polaron case is not covered by the above result and has to be discussed separately. Fr\"{o}hlich polaron Hamiltonian is indeed identified by the choices
\beq
	\label{eq:omega1}
	\omega(\kv) = 1,
\eeq
and
\beq
	\label{eq:g3}
	g(\kv; \xv) = \frac{1}{|\kv|^{\frac{d-1}{2}}} e^{- i \kv \cdot \xv},	
\eeq
which is clearly not in $ L^{\infty}(\R^d; L^2(\R^d)) $. In fact, the only way to give a meaning to the formal expression $H $ is through its quadratic form 
\beq
	\label{eq:formQ}
	Q_{H}[\Psi] : = \bra{\Psi}{H}\ket{\Psi}, 
\eeq
which can be shown to be well defined at least in a dense subset of the Hilbert space (see the Appendix). Moreover one can prove (see, \emph{e.g.}, \cite{Frank:2014aa,Griesemer:2015aa}) that the form is closable and its closure defines a unique self-adjoint operator:

	\begin{proposition}[Self-adjointness of $ H $ -- case 3.]
  		\label{pro:sa2}	
  		\mbox{}	\\
  		Let $ H $ be given by \eqref{eq:polaron ham} with interaction \eqref{eq:polaron interaction} and let the
                assumption \eqref{eq:V} be satisfied. Then the quadratic
                form $ Q_{H}[\Psi] $ is closed and bounded from below and identifies a unique self-adjoint
                operator, again denoted by $ H $, with domain
                $ \dom(H) \subset \dom(\sqrt{-\Delta+U_+}) \cap \dom(\sqrt{\mathrm{d}\Gamma(1)})$.
	\end{proposition}
	
	\begin{proof}
		Since $ H $ is defined only in the quadratic form sense, one needs to use the KLMN Theorem, in order to show that the interaction term is an infinitesimally small perturbation of $ \hamf + U_+ $. For the convenience of the reader we recall some details of the proof in the Appendix.
	\end{proof}

\subsection{Quasi-classical limit}
\label{sec:semiclassics}

We now describe in mathematical details the procedure of the quasi-classical limit. First of all we want to restrict our attention to the system of particles alone and, in order to do that, we trace out the
field's degrees of freedom. The control parameter $\varepsilon$ on the
field, which will eventually be taken to zero, is introduced through \eqref{eq:CCR} or, more precisely, as 
\begin{equation}
  	\label{eq:ccr l2}
  	\lf[a(f),a^{\dagger}(g)\ri]=\varepsilon \braket{f}{g}_{\hh},
\end{equation}
for a generic pair of functions $ f, g \in \hh $. Notice that such a choice implies that both the creation and annihilation operators are of order $ \sqrt{\eps} $. Analogously, $H$ depends on $\varepsilon$ through the field free energy
$\mathrm{d}\Gamma(\omega) $, which is of order $ \eps $, and the interaction $A(f) $ proportional to $ \sqrt{\eps} $ again. Accordingly, quantum states for the field might in
general be $\varepsilon$-dependent. In Sect. \ref{sec:introduction} we have discussed the physical meaning of the limit $\varepsilon\to 0$, that we are going to consider in the following.

The Fock partial trace of an operator (quadratic form) on
$L^2 \lf(\R^{dN} \ri)\otimes \fock(\mathfrak{H})$ is defined as follows: let $Q$ be a quadratic form on the full Hilbert space $L^2 \lf(\R^{dN} \ri)\otimes \fock(\mathfrak{H})$, which should be thought of as the quadratic form associated with the operator $ H $, and let $ \dom_0[Q] $ be a core domain for $ Q $ given by tensor product states, {\it i.e.},
\beq
	\dom_0[Q] : = \lf\{ \psi \otimes \Psi \: \Big| \: \psi \in \dom_{0,1} \subset L^2 \big(\R^{dN} \big), \Psi \in \dom_{0,2} \subset \fock(\hh) \ri\},
\eeq
where $ \dom_{0,j} $ are densely defined subspaces. Then the {\it Fock partial trace} $ \QQ $ of $ Q $ w.r.t. a field state $\Psi_{\eps} \in \fock(\hh) $ is the quadratic form on $L^2 (\mathbb{R}^{dN} )$
\beq
   \QQ[\psi] := Q[\psi \otimes\Psi_{\eps}],
\eeq
which is densely defined on $ \dom_{0,1} $. Similarly one can define the sesquilinear form $ \QQ[\psi,\phi] $ as
\beq
	\QQ[\psi,\phi] : = Q\lf[\psi \otimes \Psi_{\eps}, \phi \otimes \Psi_{\eps}\ri],
\eeq
or, equivalently, from $ \QQ[\psi] $ by polarization.

Such a procedure can be applied to the full Hamiltonian $ H $, yielding a quadratic form on $ L^2\lf(\R^{dN}\ri) $, which is associated to a Schr\"{o}dinger operator with an $ \eps$-dependent potential: 

	\begin{proposition}[Partial trace]
		\label{pro:trace}
		\mbox{}	\\
 		Let \eqref{eq:g1} (resp. \eqref{eq:g2}) be satisfied and $ Q_{H} $ be the sesquilinear form associated to the self-adjoint operator $ H $. Then the partial trace $\QQ_{H}$ of $ Q_H $ on $ \Psi_{\eps} \in \fock(\hh) $ is densely defined on $ C^{\infty}_0 \lf(\mathbb{R}^{dN} \ri)$ for any
  $ \Psi_{\eps}\in \dom  (\sqrt{\mathrm{d}\Gamma(\omega)} ) $. Moreover for any $ \psi, \phi \in C^{\infty}_0(\R^{dN}) $, the quadratic form $ \QQ_H $ is given by
  		\beq
  			\label{eq:trace}
    			\QQ_H[\psi,\phi] = \meanlr{\psi}{\HH_0 +  \sum V_{\varepsilon,\Psi_{\eps}}(\xv_j) + \ceps}{\phi}_{L^2 (\mathbb{R}^{dN})},
  		\eeq
  		where 
  		\beq
  			\label{eq:V and ceps}
  			V_{\varepsilon,\Psi_{\eps}}(\xv) = \meanlr{\Psi_{\eps}}{A(\xv)}{\Psi_{\eps}}_{\fock(L^2 (\mathbb{R}^{d}))},	\qquad  \ceps = \meanlr{\Psi_{\eps}}{\diff\Gamma(\omega)}{\Psi_{\eps}}_{\fock(L^2 (\mathbb{R}^{d}))}.
		\eeq
	\end{proposition}
	
	\begin{proof}
		The result is obtained by computing the partial trace in a straightforward way. The well-posedness of the r.h.s. on smooth functions with compact support is inherited from the properties of the quadratic form $ Q_H $, whose domain contains such type of wave functions for the particle subsystem.
	\end{proof}
	
	 \begin{remark}[Partial trace for the polaron]
	 	\label{rem:trace polaron}
	 	\mbox{}	\\
	 	The above Proposition does not apply straightforwardly to the polaron model, since by \eqref{eq:g3} $ \lambda \notin L^2(\R^d) $. It is however possible to prove an analogous statement where the main difference is that $ \QQ_H $ is only a quadratic form and the association to the operator on the r.h.s. purely formal, until one proves that such a form is closed and defines a unique self-adjoint operator (see the Appendix).
	 \end{remark}
	 
In general it is very difficult to characterize the effective potential $ V_{\varepsilon,\Psi_{\eps}} $ obtained in this
way. For example, it is not a priori assured that the r.h.s. of \eqref{eq:trace} is a sesquilinear form associated to a
unique self-adjoint operator, even if $H$ is self-adjoint. Conversely, it might happen that such a form identifies a unique self-adjoint operator, even though $ H$ is not self-adjoint. Such problems however do not show up in the limit $ \eps \to 0 $, if reasonable assumptions on the state $ \Psi_{\eps} $ are made.

The result below is based on the techniques of semiclassical analysis for infinite dimensional systems introduced in 
\cite{ammari:nier:2008,MR2513969,MR2802894,2011arXiv1111.5918A}. Let us recall that both the operators, {\it e.g.}, $\mathrm{d}\Gamma(\omega)$, and vectors, {\it e.g.}, $\Psi_\eps $, in the Fock space depend on $\varepsilon$.

	\begin{proposition}[Classical limit]
		\label{pro:classical limit}
		\mbox{}	\\
  		Let $ \Psi_{\eps} \in  \fock(\mathfrak{H})$ be such that, uniformly in $\varepsilon $ small,
 		 \begin{itemize}
  			\item there exist $\delta \geq \frac{1}{2} $ and $ C < + \infty $ such that
    				\begin{equation}
      				\meanlr{\Psi_\varepsilon}{\lf( \mathrm{d}\Gamma(1) \ri)^{\delta}}{\Psi_\varepsilon}_{\fock(\hh)} \leq  C;
    				\end{equation}
  			\item there exists $ C' < + \infty $, such that
    				\begin{equation}
     					\mean{\Psi_\varepsilon}{\mathrm{d}\Gamma(\omega)}{\Psi_\varepsilon}_{\fock(\hh)} \leq C',
    				\end{equation}
    				where $ \omega $ is the multiplication operator by the function $ \omega_{\nv} $ or $ \omega(\kv) $.
  		\end{itemize}
  		Then there is a subsequence $ \lf\{ \Psi_{\eps_k} \ri\}_{k\in \mathbb{N}} $, $\varepsilon_k \to 0$, as $ k \to \infty $, and a probability measure $\mu \in \M(\mathfrak{H}) $, such that:
  		\begin{itemize}
  			\item $\mu$ is concentrated on $ \dom(\omega)$;
  			\item $\lVert z \rVert_{\mathfrak{H}}^{\alpha_1}$ and $\lVert \sqrt{\omega} z\rVert_{\mathfrak{H}}^{\alpha_2}$, with $\alpha_1\leq 2\delta$ and $\alpha_2\leq 2$, are integrable with respect to the measure $\mathrm{d} \mu(z)$ and
    				\begin{equation}
     					\label{eq:field energy limit}
     					\lim_{k\to \infty} \mean{\Psi_{\varepsilon_k}}{\diff \Gamma(\omega)}{\Psi_{\varepsilon_k}}_{\fock(\hh)}=\int_{\mathfrak{H}}^{}\mathrm{d}\mu(z) \: \lf\|\sqrt{\omega} z \ri\|_{\hh}^2;
   				\end{equation}
 			 \item for any $ g \in \mathfrak{H} $,
    				\begin{equation}
     					\label{eq:A limit}
      				\lim_{k\to \infty} \mean{\Psi_{\varepsilon_k}}{a(g) + a^{\dagger}(g) }{\Psi_{\varepsilon_k}}_{\fock(\hh)} =  2\Re \int_{\mathfrak{H}}^{}  \mathrm{d}\mu(z) \: \braket{z}{g}_{\mathfrak{H}}.
    				\end{equation}
   		\end{itemize}
	\end{proposition}
	
	\begin{proof}
                The existence of a subsequence converging to the classical measure is proved in               \cite[Theorem 6.2]{ammari:nier:2008}, as well as the integrability of
                $\lVert z \rVert_{\mathfrak{H}}^{\alpha_1}$, $\alpha_1\leq 2\delta$. The concentration of
                $\mu$ in $D(\sqrt{\omega})$, and the integrability of $\lVert \sqrt{\omega} z \rVert_{\mathfrak{H}}^{\alpha_2}$,
                $\alpha_2\leq 2$, as well as the convergence of the corresponding evaluation of
                $\mathrm{d}\Gamma(\omega)$ is proved in \cite[Lemma 3.13]{2011arXiv1111.5918A}. The
                convergence of the expectation of the field operator also follows along the same guidelines,
                an interested reader might consult, \emph{e.g.}, \cite{Ammari:2014aa}.
	\end{proof}
	
	In the following, the role of the function $ g $ will be played by the cut-off $ \lambda \in \hh $, so that \eqref{eq:A limit} will allow us to take the limit $ \eps \to 0 $ of the interaction term in the Hamiltonian $ H $. Once again, the case of the polaron is excluded since $ \lambda \notin L^2(\R^d) $ and therefore a comment is in order.

	\begin{remark}[Classical limit for the polaron]
		\label{rem:classical limit polaron}
		\mbox{}	\\
		The limit \eqref{eq:A limit} may also hold true for functions $ g $ not belonging to $ \hh =
                L^2(\R^d) $. The easiest situation is given by a function $g$ that is in $H^{-s}(\R^d)$,
                $s>0$. In this case, the scalar product $ \braket{\: \cdot \:}{\: \cdot \:}_2 $ defined on
                $H^s \otimes H^s$ extends to a continuous duality map on $ H^s \otimes H^{-s} $ (denoted by
                $ \braket{\: \cdot \:}{\: \cdot \:}_{*} $).  Hence for any $ g \in H^{-s}(\R^d) $, \eqref{eq:A limit} is
                reformulated as follows:
		\begin{equation}
     			 \label{eq:A limit polaron}
     			 \lim_{k\to \infty} \meanlr{\Psi_{\varepsilon_k}}{a(g) + a^{\dagger}(g) }{\Psi_{\varepsilon_k}}_{\fock(L^2)} =  2\Re \int_{L^2(\R^d)}^{}  \mathrm{d}\mu(z) \: \braket{z}{g}_{*},
    		\end{equation}
    		where the r.h.s. is finite if and only if $\mu$ is concentrated on $H^s(\R^d) $.

                Another important example is given by generalized functions $\xv \mapsto g(\xv)$ whose inverse
                derivative (more precisely $ (- \Delta + 1)^{-1/2} g $) takes values in $\hh$ for a.e. $\xv \in \R^{d} $. We
                adopt the natural notation $W^{-1,\infty}\lf(\R^{d},\hh\ri)$ for the space of such functions. Now for
                any $g(\xv)\in W^{-1,\infty}\lf(\R^{d},\hh\ri)$, the convergence
                \begin{equation}
                  \label{eq:A limit polaron 2}
                  \lim_{k\to \infty} \sum_{j=1}^N\meanlr{\Psi_{\varepsilon_k}}{a(g(\xv_j)) + a^{\dagger}(g(\xv_j)) }{\Psi_{\varepsilon_k}}_{\fock(\hh)} =  2\Re \sum_{j=1}^N\int_{\hh}^{}  \mathrm{d}\mu(z) \: \braket{z}{g(\xv_j)}_{\hh}
                \end{equation}
                has to be interpreted as the convergence in a dense domain of quadratic forms in $L^2 \lf(\R^{dN}\ri)$,
                and the limit defines a quadratic form bounded by $ \QQ_{\sqrt{-\Delta}}$. In fact, since $g(\xv)\in W^{-1,\infty}(\R^{d},\hh)$, there exists a
                $\tilde{\gv}(\xv) = \lf(\tilde{g}_1(\xv), \ldots, \tilde{g}_d(\xv) \ri) \in L^{\infty}(\R^{d},\hh\otimes \R^{d})$ such that
                \begin{equation*}
                  g(\xv)= \lf[-i\nabla, \:\tilde{\gv}(\xv) \ri].
                \end{equation*}
                Therefore it follows that
                \begin{equation*}
                  2\Re \sum_{j=1}^N\int_{\hh}^{}  \mathrm{d}\mu(z) \: \braket{z}{g(\xv_j)}_{\hh}=\sum_{j=1}^N  \lf[ -i\nabla_{\xv_j}, \: 2\Re \int_{\hh}^{}  \mathrm{d}\mu(z) \: \braket{z}{\tilde{\gv}(\xv_j)}_{\hh} \ri]\; .
                \end{equation*}
	\end{remark}
	
	An important feature which we have already commented upon extensively in Sect.~\ref{sec:main} is the fact that, given any $ \Psi_{\eps}$ satisfying the hypothesis of Proposition \ref{pro:classical limit}, there exists at least one limit measure $ \mu $, which might depend on the chosen subsequence. When we say that, as in \eqref{eq:eps convergence},
	\bdm
		\Psi_{\eps} \xrightarrow[\eps \to 0]{} \mu,
	\edm
	we mean that the subsequence has been chosen (and thus the limit point $ \mu $) or the limit is unique and no choice has to be made. To ensure that the results of Proposition \ref{pro:classical limit} hold true, we will also assume that \eqref{eq:state} hold true, {\it i.e.},
	\beq
		\tag{{\bf A}3}
		\mean{\Psi_\varepsilon}{\mathrm{d}\Gamma(1)}{\Psi_\varepsilon}_{\fock(\hh)} \leq C < + \infty; 	\qquad		\mean{\Psi_\varepsilon}{\mathrm{d}\Gamma(\omega)}{\Psi_\varepsilon}_{\fock(\hh)} \leq C' < +\infty.
	\eeq
	
	A first important consequence of Proposition \ref{pro:classical limit} is the following.
	
	\begin{corollary}[Pointwise convergence]
		\label{cor:pointwise}
		\mbox{}	\\
		Let the assumptions of Proposition \ref{pro:classical limit} be satisfied, {\it i.e.}, $ \Psi_{\eps} \to \mu $ as $ \eps \to 0 $ in the sense of \eqref{eq:eps convergence}, and let $ g(\xv) \in L^{\infty}\lf(\R^{dN}; \hh\ri) $. Then
		\beq
			\label{eq:pointwise}
			V_{\eps, \Psi_{\eps}}(\xv)  \xrightarrow[\eps \to 0]{\mathrm{a.e.}} V_{\mu}(\xv) = 2\Re \int_{\mathfrak{H}}^{}  \mathrm{d}\mu(z) \: \braket{z}{g(\xv)}_{\hh}.
		\eeq
	\end{corollary}
	
	\begin{remark}[Convergence along a subsequence]
		\label{rem:subsequence}
		\mbox{}	\\
		The convergence $ \Psi_{\eps} \to \mu $ is meant on a specific subsequence, if the limit is not unique. Therefore the above pointwise limit \eqref{eq:pointwise} holds true along the same subsequence and, in order to be precise, we should have stated the convergence of $ V_{\eps_k,\Psi_{\eps_k}} $, as $ k \to \infty $. However, we choose not to use such a cumbersome notation but we stress that \eqref{eq:pointwise} should be taken in the appropriate sense.
	\end{remark}
	
	\begin{proof}
		By treating $ \xv \in \R^{dN} $ as a parameter, one can directly apply Proposition \ref{pro:classical limit} and specifically \eqref{eq:A limit}: for a.e. $ \xv \in \R^{dN} $, $ g(\xv) $ belongs to $ \hh $ and therefore 
		\bdm
			\meanlr{\Psi_{\varepsilon_k}}{a(g(\xv)) + a^{\dagger}(g(\xv)) }{\Psi_{\varepsilon_k}}_{\fock(\hh)} \xrightarrow[\eps \to 0]{}  2\Re \int_{\mathfrak{H}}^{}  \mathrm{d}\mu(z) \: \braket{z}{g(\xv)}_{\mathfrak{H}}, 
		\edm
		where the convergence is meant on the chosen subsequence.
	\end{proof}
	
	We conclude with an obvious consequence of Proposition \ref{pro:classical limit} and assumptions \eqref{eq:state}:
	
	\begin{corollary}[Field energy]
		\label{cor:field energy}
 		\mbox{}	\\
 		Let the assumptions \eqref{eq:state} be satisfied, then	
  		\begin{equation}
  			\label{eq:lim field energy}
   			 \lim_{\varepsilon\to 0}c_{\varepsilon}=\int_{\mathfrak{H}}^{}  \mathrm{d}\mu(z) \: \lf\| \sqrt{\omega} z \ri\|_{\mathfrak{H}}^2 < + \infty.
  		\end{equation}
	\end{corollary}
		
	Now that we have specified the key mathematical tools of our analysis, we proceed with the proofs of the results stated in Sect. \ref{sec:main}. 

\subsection{Discrete modes}

We aim at proving Theorem \ref{teo:conv1}: the key ingredient is the convergence guaranteed by Proposition \ref{pro:classical limit}. The other properties can be proven by direct inspection. We recall that the full Hamiltonian $ H $ is given in \eqref{eq:full ham} with interaction \eqref{eq:interaction}. We denote by $ C_{\mathrm{b}}\lf(\R^{dN}\ri) $ the space of bounded continuous functions on $ \R^{dN} $, while $ C_{\infty}\lf(\R^{dN}\ri) $ stands for continuous functions vanishing as $ |\xv| \to \infty $.

	Before attacking the proof of Theorem \ref{teo:conv1}, we only need one more technical result:

	\begin{lemma}
  		\label{lemma:pot bounded}
  		\mbox{}	\\
  		Let the assumption \eqref{eq:g1} be satisfied and $ \Psi_\varepsilon \in \dom(\mathrm{d}\Gamma(1)^{1/4}) $ uniformly in $ \eps $, then $ V_{\varepsilon,\Psi_\varepsilon}(\xv) \in C_{\mathrm{b}}(\mathbb{R}^d)$, {\it i.e.},
  		\beq
  			\label{eq:pot bounded}
  			\sup_{\xv \in \R^d} \lf| V_{\varepsilon,\Psi_\varepsilon}(\xv) \ri| \leq C < +\infty,
		\eeq
		uniformly in $ \eps $.
	\end{lemma}
	\begin{proof}
		The key observation is that the following bound holds true:
  		\beq
  			\label{eq:fock est form}
    			\lf| \meanlr{\Psi_{\eps}}{A(\xv)}{\Psi_{\eps}}_{\fock} \ri| \leq 2 \lf\| g_{\nv}(\xv)\ri\|_{\ell^2} \lf\| \bigl(\mathrm{d}\Gamma(1)+1\bigr)^{1/4}\Psi_\varepsilon  \ri\|_{\fock}^2.
  		\eeq
  		Indeed, it yields
  		\bmln{
     			\sup_{\xv \in \mathbb{R}^d}  \lf| V_{\varepsilon,\Psi_\varepsilon}(\xv)  \ri| = \sup_{\xv \in \mathbb{R}^d} \lf| \meanlr{\Psi_{\eps}}{A(\xv)}{\Psi_{\eps}}_{\fock} \ri| \leq C  \lf\| \bigl(\mathrm{d}\Gamma(1)+1\bigr)^{1/4}\Psi_\varepsilon  \ri\|_{\fock}^2 \\
     			\leq C \lf[ \lf\| \mathrm{d}\Gamma(1)^{1/4}\Psi_\varepsilon  \ri\|_{\fock}^2 + \lf\| \Psi_\varepsilon  \ri\|_{\fock}^2 \ri] \leq C,
    		}
    		and therefore $ V_{\varepsilon,\Psi_\varepsilon}(\xv) $ is uniformly bounded in $ \eps $. To prove continuity we use again \eqref{eq:fock est form}:
  		\bdm
  		     \lf| V_{\varepsilon,\Psi_\varepsilon}(\xv) - V_{\varepsilon,\Psi_\varepsilon}(\yv) \ri|  \leq C  \lf\| \bigl(\mathrm{d}\Gamma(1)+1\bigr)^{1/4}\Psi_\varepsilon  \ri\|_{\fock}^2 \lf\| g_{\nv}(\xv) - g_{\nv}(\yv) \ri\|_{\ell^2} \xrightarrow[\xv \to \yv]{} 0,
    		\edm
    		by dominated convergence.
  		
  		We prove now \eqref{eq:fock est form}: let $\Psi\in \mathscr{D}\bigl(\mathrm{d}\Gamma(1)^{1/4}\bigr)$ and $ g \in \ell^2 \lf( \R^{dN} \ri)$, then, using Cauchy-Schwarz twice, one has
                \bmln{
                	\lf| \braket{\Psi}{a(g)\Psi}_{\fock} \ri| \leq \sum_{m = 0}^{\infty}\sqrt{\eps m+1}\:\bigg| \sum_{\mathbf{n}\in \mathbb{Z}^d}^{} \Psi_{m}^{*} (\mathbf{n}_1, \ldots,\mathbf{n}_m) g_{\mathbf{n}} \Psi_{m+1}(\mathbf{n},\mathbf{n}_1, \ldots ,\mathbf{n}_m)  \bigg| \\
                	\leq \sum_{m = 0}^{\infty} \lf\| g  \ri\|_{\ell^2}^{} \lf\| (\eps m+1)^{1/4}\Psi_m  \ri\|_{\ell^2_m}^{} \lf\| (\eps m+1)^{1/4}\Psi_{m+1}  \ri\|_{\ell^2_{m+1}}^{}\\\leq \lf\| g  \ri\|_{\ell^2} \lf\| \lf(\mathrm{d}\Gamma(1)+1 \ri)^{1/4}\Psi  \ri\|_{\fock}^{} \lf\|\mathrm{d}\Gamma(1)^{1/4}\Psi  \ri\|_{\fock}^{}\; ,
                 }
                where for any $m\in \N \cup \{0 \} $, $\Psi_m\in \ell^2(\mathbb{Z}^d)^{\otimes _{\mathrm{sym}} m}= : \ell^2_m$ is the component of
                $\Psi$ with $m$ modes, {\it i.e.}, $ \Psi = \lf( \Psi_0, \ldots, \Psi_m, \ldots \ri) \in \fock\lf(\ell^2\lf(\Z^{d}\ri)\ri)$.
	\end{proof}
	
	We are now in position to complete the proof of the main result about the model with discrete modes of radiation.
	
	\begin{proof}[Proof of Theorem \ref{teo:conv1}]
		Under the hypothesis of Theorem \ref{teo:conv1} and thanks to Lemma \ref{lemma:pot bounded}, $ V_{\eps,\Psi_{\eps}} $ is an infinitesimally small perturbation of $ \HH_0 $ in the sense of Kato. Therefore $ \HH_{\eps} $ is self-adjoint on the domain of self-adjointness $ \dom(\HH_0) $ of $ \HH_0 $. 
		
		Moreover one has
		\beq
			\label{eq:pot mu bounded}
			\sum_{j = 1}^N \sup_{\xv_j \in \R^{d}} \lf| V_{\mu}(\xv_j) \ri| \leq 2N  \bigg( \sum_{\nv \in \Z^d} \lf| \lambda_{\nv} \ri|^2 \bigg)^{1/2} \bigg( \int_{\ell^2(\Z^d)}^{}  \mathrm{d}\mu(z) \: \lf\| z \ri\|^2_{\ell^2(\Z^d)} \bigg)^{1/2} < +\infty,
		\eeq
		thanks to the assumptions \eqref{eq:g1} on $ \lambda_{\nv} $ and \eqref{eq:state} on $ \Psi_{\eps} $ in combination with Proposition \ref{pro:classical limit}. Therefore $ \HHe $ is also self-adjoint on $ \dom(H_0) $.
		
		To prove the convergence in norm resolvent sense, pick any $ {\zeta} \in \rho(\HH_{\eps}) \cap \rho(\HHe) $ uniformly in $ \eps $, {\it i.e.}, such that there exists $ C > 0 $ so that $ \dist({\zeta}, \sigma(\HH_{\eps})) > C $. Then by the second resolvent identity 
		\bml{
			\sup_{\lf\| \psi \ri\|_2 = 1} \lf\| {\lf[ \lf( \HH_{\eps} - \zeta \ri)^{-1}  - \lf( \HHe - \zeta \ri)^{-1} \ri]} \psi \ri\|_{L^2} \\
			\leq \sup_{\lf\| \psi \ri\|_2 = 1}  \sum_{j = 1}^N \lf\| \lf(\HH_{\eps} - {\zeta} \ri)^{-1} \lf( V_{\eps, \Psi_{\eps}}(\xv_j) - V_\mu(\xv_j) \ri) \lf(\HHe - {\zeta} \ri)^{-1} \psi \ri\|^2_{L^2} \\
			\leq C \sup_{\lf\| \psi \ri\|_2 = 1} \sum_{j = 1}^N \lf\|  \lf( V_{\eps, \Psi_{\eps}}(\xv_j) - V_\mu(\xv_j) \ri) \lf(\HHe - {\zeta} \ri)^{-1} \psi \ri\|^2_{L^2} \leq  C N \sup_{\lf\| \psi \ri\|_2 = 1} \lf\|  \lf(\HHe - {\zeta} \ri)^{-1} \psi \ri\|^2_{L^2} \\
			\leq C \sup_{\lf\| \psi \ri\|_2 = 1} \lf\| \psi \ri\|^2_{L^2} \leq C < +\infty,
		}
		by the uniform boundedness of $ V_{\eps,\Psi_{\eps}} $ and $ V_\mu $ proven in \eqref{eq:pot mu bounded} and in \eqref{eq:pot bounded} and the assumptions on $ {\zeta} $. Therefore the integrand on the l.h.s. is uniformly bounded by a $ L^1 $ function, whose norm is finite. Hence we can apply a dominated convergence argument and the result then follows from pointwise convergence of $ V_{\eps, \Psi_{\eps}} $ to $ V_\mu $ proven in Corollary \ref{cor:pointwise}.
	\end{proof}

\subsection{Nelson model and trapping potentials}

The proof of Theorem \ref{teo:conv2} is a trivial adaptation of the proof of Theorem \ref{teo:conv1} discussed in the previous Sect.: it is indeed sufficient to replace $ \ell^2(\Z^d) $ with $ L^2(\R^d) $ and follow step by step the same arguments. We omit the details.

We turn now our attention to the result presented in Sect. \ref{sec:traps} and specifically Theorem \ref{teo:trap}. We recall that the setting is slightly different: the goal is to derive the effective particle Hamiltonian under restrictive assumptions on the field state, which is assumed to be a squeezed coherent state, {\it i.e.}, a state of the form \eqref{eq:coherent state},
\bdm
	\Xi(f) =  \WW\lf( \tx\frac{1}{i\eps} f \ri)\Omega,
\edm
$ \Omega $ being the vacuum state. More precisely we assume that $ \Psi_{\eps} $ is given by \eqref{eq:coherent trap}, {\it i.e.},
\bdm
	\Xi(f_{W,\varepsilon}) := \WW\lf( \tx\frac{1}{i\eps} f_{W,\eps} \ri)\Omega,
\edm
where 
\bdm
 	f_{W,\varepsilon}(\kv)= \frac{1}{2(2\pi)^{d/2} \lambda^{*}(\kv)} \lf(  \widehat{\varphi_{\varepsilon} *  W} \ri)(\kv),
\edm
for $ W \in L^2_{\mathrm{loc}}(\R^d; \R^+) $ and a suitable mollifier $ \varphi_{\varepsilon}(\xv)=\varepsilon^{-d} \varphi(\xv/\varepsilon) $, $ \varphi \in C^{\infty}_0(\R^d) $ with $ \lf\| \varphi \ri\|_1 = 1 $. Note in particular that we drop in this Sect. the assumptions \eqref{eq:state} on the field state and therefore Proposition \ref{pro:classical limit} does not apply. In the case of coherent states however the derivation of the effective potential is much more explicit and there is no need to pass through the convergence to classical measures:

	\begin{proposition}[Classical limit of coherent states]
		\label{pro:coherent states}
		\mbox{}	\\
		Let $ \Psi_{\eps} = \Xi(f_{\eps}) $ be a state of the form \eqref{eq:coherent state} for some $ f_{\eps} \in L^2(\R^d) $, for any $ \eps < 1 $. Then
                \begin{align}
                  	\label{eq:coherent potential}
			&\meanlr{\Psi_{\eps}}{A(\xv)}{\Psi_{\eps}} = 2\Re \int_{\R^d} \diff \kv \: e^{i \kv \cdot \xv} f_{\varepsilon}(\kv) \lambda^{*}(\kv)\; .
                      \end{align}
                      If in addition $\omega f_{\varepsilon}\in L^2 (\R^d  )$, then
                      \begin{align}
                        \label{eq:coherent energy}
                        &\meanlr{\Psi_{\eps}}{\mathrm{d}\Gamma(\omega)}{\Psi_{\eps}} = \int_{\R^d} \diff \kv \: \omega(\kv) \lf| f_{\varepsilon}(\kv) \ri|^2.
                      \end{align}
	\end{proposition}
	
	\begin{proof}
          The result is a consequence of the following well known property of the Weyl operators: for any
          $g\in L^2 (\R^d )$
          \begin{align*}
            &\WW^{\dagger}\lf( \tx\frac{1}{i\eps} g \ri)a(\kv) \WW\lf( \tx\frac{1}{i\eps} g \ri)=a(\kv)+ g(\kv)\; , \\
            &\WW^{\dagger}\lf( \tx\frac{1}{i\eps} g \ri)a^{\dagger}(\kv) \WW\lf( \tx\frac{1}{i\eps} g \ri)=a^{\dagger}(\kv)+g^{*}(\kv)\;.
          \end{align*}
          It then follows that
          \begin{equation*}
            \begin{split}
              \meanlr{\Psi_{\eps}}{A(\xv)}{\Psi_{\eps}}_{\fock} = 2\Re \int_{\R^d} \diff \kv \: e^{i \kv \cdot \xv} f_{\varepsilon}(\kv) \lambda^{*}(\kv)+ \meanlr{\Omega}{A(\xv)}{\Omega}_{\fock},
            \end{split}
          \end{equation*}
          and the second term in the right hand side is zero since it consists of the action of the
          annihilation operator on the vacuum (once on the right and once on the left). Analogously,
          \bmln{
              \meanlr{\Psi_{\eps}}{\mathrm{d}\Gamma(\omega)}{\Psi_{\eps}}_{\fock} = \int_{\R^d} \diff \kv \: \omega(\kv) f_{\varepsilon}^{*}(\kv) f_{\varepsilon}(\kv)+\meanlr{\Omega}{\mathrm{d}\Gamma(\omega)}{\Omega}_{\fock} \\
              + 2\Re\meanlr{\Omega}{a(\omega f_{\varepsilon})}{\Omega}_{\fock},
          }
          and again the second and third term on the right hand side vanish because the annihilation operator
          acts on the vacuum.
	\end{proof}

	\begin{remark}[Convergence to a classical measure]
		\label{rem:cylindrical measure}
		\mbox{}	\\
		One can naturally wonder whether a result like the one stated in Proposition \ref{pro:classical limit}, {\it i.e.}, a sort of convergence of $ \Psi_{\eps} $ to a classical measure, holds true also for state of the form \eqref{eq:coherent trap}, or, more generally, for $ \Xi(f_{\eps}) $. The answer is actually given by \cite[Theorem 3.15]{Falconi:2016ab}: there is always at least one cluster point, but unfortunately such point might be a {\it cylindrical measure} instead of a true measure. Given the properties of cylindrical measures, this actually means that, by suitably enlarging the space, one can make the limit point $ \mu $ a true measure, but the key feature is that typically the support of $ \mu $ is {\it outside} $ \hh $.
	\end{remark}
	
	A technical but useful result is the following
	
	\begin{lemma}
		\label{lemma:l2 convergence}
		\mbox{}	\\
		For any $ \varphi \in C^{\infty}_0(\R^d) $ with $ \lf\| \varphi \ri\|_1 = 1 $ and $ W \in L^2_{\mathrm{loc}}(\R^d) $,
		\beq
			\label{eq:l2 convergence}
			\varphi_{\eps} * W \xrightarrow[\eps \to 0]{L^2_{\mathrm{loc}}(\R^d)} W,
		\eeq
		where $ \varphi_{\eps}(\xv) : = \eps^{-d} \varphi(\xv/\eps) $.
	\end{lemma}
	
	\begin{proof}
		We use once more dominated convergence: let $ K \subset \R^d $ be any compact set, then
		\bmln{
			\int_{K} \diff \xv \: \lf| \lf( \varphi_{\eps} * W \ri)(\xv) - W(\xv) \ri|^2 = \int_{K} \diff \xv \: \bigg| \int_{\supp(\varphi)} \diff \xvp \: \varphi(\xvp) \lf( W(\xv + \eps \xvp) - W(\xv) \ri) \bigg|^2	\\
			\leq \lf\| \varphi \ri\|^2_{L^2(\R^d)} \int_{K} \diff \xv \: \int_{\supp(\varphi)} \diff \xvp \: \lf| W(\xv + \eps \xvp) - W(\xv) \ri|^2 \\
			\leq C \bigg[ \int_{K} \diff \xv \: \int_{\supp(\varphi)} \diff \xvp \: \lf| W(\xv + \eps \xvp) \ri|^2 + |\supp \varphi| \lf\| W \ri\|^2_{L^2(K)} \bigg] \leq C_K,
		}
		so that we can take the limit $ \eps \to 0 $ inside the integral by Vitali's Theorem. Since $ \varphi_{\eps} * W $ converges a.e. to $ W $ on any compact set, we obtain the result.
	\end{proof}
	
	The last technical ingredient for the proof of Theorem \ref{teo:trap} is stated in the next Lemma.
	
	\begin{lemma}
		\label{lemma:dprime}
		\mbox{}	\\
  		Let $ T_{\varepsilon} : =-\Delta + U + V_{\varepsilon} $  be a family of self-adjoint operators on $L^2 (\mathbb{R}^{dN} )$ such that \eqref{eq:V} is satisfied and $V_{\varepsilon} \in L^2_{\mathrm{loc}} \lf(\mathbb{R}^{dN}; \mathbb{R}_+ \ri) $, with {$L^2 $ norm uniformly bounded in $ \eps $ in any compact set}, and
  		\beq
  			\label{eq:V conv}
  			V_{\eps}(\xv_1, \ldots, \xv_N)  \xrightarrow[\eps \to 0]{L^2_{\mathrm{loc}}(\R^{dN}; \R^+)} V_0(\xv_1, \ldots, \xv_N).	
		\eeq
		Then $T_{\varepsilon} \xrightarrow[\eps \to 0]{} T=-\Delta+U+V_0$ in strong resolvent sense.
	\end{lemma}
	
	\begin{proof}	
		The result is a direct consequence of a general result about convergence of operators (see, {\it e.g.}, \cite[Theorem VIII.25]{MR0493419}): if there exists a common core for all the operators $ T_{\eps} $, $ T_0 $ and on that core $ T_{\eps} \psi \to T_0 \psi $, then the operators converge in strong resolvent sense. 
		
		Under the hypothesis of the Lemma both the sequence of operators $ T_{\eps} $ and $ T_0 $ are essentially self-adjoint on $ C^{\infty}_0\lf(\R^{dN}\ri) $ (see, {\it e.g.}, \cite[Theorem X.28]{RS2}). Moreover, for any $ \psi \in C^{\infty}_0\lf(\R^{dN}\ri) $,
		\bmln{
			\lf\| \lf( T_{\eps} - T_0 \ri) \psi \ri\|^2_{L^2(\R^{dN})} = \lf\| \lf( V_{\eps} - V_0 \ri) \psi \ri\|^2_{L^2(\R^{dN})} = \int_{\R^{dN}} \diff \xv \: \lf( V_{\eps} - V_0 \ri)^2 \lf| \psi \ri|^2 \\
			\leq C \lf\| V_{\eps} - V_0 \ri\|_{L^2(\supp({\psi}))} \lf( \lf\| V_{\eps} \ri\|_{L^2(\supp({\psi}))} + \lf\| V_0 \ri\|_{L^2(\supp({\psi}))} \ri) \xrightarrow[\eps \to 0]{} 0,
		}
		since $ |\psi|^2 $ is bounded and has compact support.
	\end{proof}
	
	We proceed now with the proof of the main result:
	
	\begin{proof}[Proof of Theorem \ref{teo:trap}]
		Thanks to \eqref{eq:coherent potential} proven in Proposition \ref{pro:coherent states}, we know that the effective potential generated by the partial trace of the field operator on coherent states of the form \eqref{eq:coherent trap} is
		\bdm
			\sum_{j = 1}^N W_{\eps}(\xv_j) = \sum_{j=1}^N \lf( \varphi_{\eps} * W \ri)(\xv_j) \in L^2_{\mathrm{loc}}\big(\R^{dN}\big).
		\edm
		Note that Proposition \ref{pro:coherent states} can be applied since $ \varphi_{\eps}*W $ is a smooth function for any $ 0 < \eps < 1 $ and therefore its Fourier transform is rapidly decaying (faster than polynomially). Hence $ f_{W,\eps} \in L^2(\R^d) $, since $ \lambda(\kv) $ diverges as $ |\kv| \to \infty $ at most polynomially.
		
		Self-adjointness of $ \HH_{\eps} : = \HH_0 + U + \sum W_{\eps} $ and $ \HHe : = \HH_{\eps} : = \HH_0 + U + \sum W $ follows, {\it e.g.}, from \cite[Theorem X.28]{RS2}, which also guarantees that $ C^{\infty}_0\lf(\R^{dN}\ri) $ is a common core for both operators. Then the combination of Lemma \ref{lemma:l2 convergence} and Lemma \ref{lemma:dprime} completes the proof.	
	\end{proof}

\subsection{Polaron}

The full Hamiltonian of Fr\"{o}hlich polaron in given in \eqref{eq:polaron ham}, although that expression is purely formal. As anticipated in Proposition \ref{pro:sa2}	and proven in the Appendix, indeed, the interaction in \eqref{eq:polaron ham} makes sense only when written as a quadratic form, which can be shown to be a small perturbation of the free quadratic form associated to $ H_{\mathrm{free}} $. Therefore we think of $ H $ as the unique self-adjoint operator associated to the quadratic form $ Q_H[\Psi]  = \bra{\Psi}{H}\ket{\Psi} $.

The reader should also keep in mind that, as discussed in Remark \ref{rem:classical limit polaron}, the convergence of quantum expectation values of creation and annihilation operators to suitable classical quantities should be taken with care, and should be interpreted as the convergence of quadratic forms.

Therefore the proof strategy has to be suitably tuned to take into account two technical features, which are specific of the polaron: on the one hand one has to switch from the outset from  Schr\"{o}dinger operators to the associated quadratic forms, and, on the other, find an alternative route, which does not require $ g(\xv) $ to be in $L^{\infty}(\R^{dN},\hh)$.

We start by showing that both $ \HH_{\eps} $, for any $ \eps $, and $ \HHe $ are self-adjoint on suitable domains:

	\begin{lemma}
  		\label{lemma:sa hheps}
  		\mbox{}	\\
  		Let the assumptions \eqref{eq:V}, \eqref{eq:omega1} and \eqref{eq:g3} be satisfied and let $ \Psi_{\eps} \in \dom(\sqrt{\mathrm{d}\Gamma(1)}) $. Then $ \HH_{\eps} $ is self-adjoint on a domain $ \dom(\HH_{\eps}) \subset \dom\lf(\sqrt{-\Delta+U_+}\ri)$ and bounded from below for any $ \eps $ small.
	\end{lemma}

	\begin{proof}
		The result is a consequence of the estimate \eqref{eq:klmn bound} proven in the Appendix and used in the application of KLMN Theorem to the quadratic form $ \Q_H $ associated to the full polaron Hamiltonian. As explained in details in the Appendix, the trick is to split the expectation value of the field operator into an infrared contribution for $ |\kv| \leq \varrho $ and an ultraviolet one for $ |\kv| \geq \varrho $, where $ \varrho > 0 $ is a positive parameter to be optimized over.
		
  		Let then $ \psi \in \dom\lf(\sqrt{-\Delta+U_+}\ri)$ and $ \Psi_{\eps} $ be normalized. Then by definition of partial trace
  		\bdm
    			\sum_{j = 1}^N \meanlr{\psi}{V_{\eps,\Psi_{\eps}}(\xv_j)}{\psi}_{L^2(\R^{dN})} =\sum_{j = 1}^N \meanlr{\psi \otimes \Psi_{\eps}}{A(\xv_j)}{\psi \otimes \Psi_{\eps}}_{L^2 \otimes \fock}.
 		\edm
		The KLMN estimate \eqref{eq:klmn bound} yields ($ - \Delta = \sum - \Delta_j $ stands here for the Laplacian on $ \R^{dN} $)
  		\bmln{
      		\bigg| \sum_{j = 1}^N \meanlr{\psi}{V_{\eps,\Psi_{\eps}}(\xv_j)}{\psi}_{L^2(\R^{dN})}  \bigg| \leq \tx\frac{1}{2}  \meanlr{\psi}{- \Delta + U_+)}{\psi}_{L^2(\R^{dN})} \\+  \lf[ \tx\frac{1}{2} \meanlr{\Psi_{\eps}}{\diff \Gamma(1)}{\Psi_{\eps}}_{\fock} + C \ri] \lf\| \psi \ri\|_{L^2}^2;
    		}	
    		where $ C $ is a finite quantity. Hence the potential 
    		\bdm
    			\sum_{j=1}^N  V_{\eps,\Psi_{\eps}}(\xv_j) + U_{\ll}(\xv_1, \ldots, \xv_N) 
		\edm
		is a small perturbation of $ Q_{- \Delta + U_+} $ in the sense of quadratic forms, since by hypothesis $ U_{\ll}  $ is infinitesimally small w.r.t. $ - \Delta + U_+ $ and therefore the relative bound can be obtained as small as $ \frac{1}{2} + \epsilon  <  1 $. The KLMN Theorem (see, {\it e.g.}, \cite[Theorem X.17]{RS2}) then yields the results.
	\end{proof}
	
	\begin{lemma}
		\label{lemma:sa hhe}
		\mbox{}	\\
		Let $\mu \in \M\lf(L^2(\R^d)\ri) $ be a probability measure, satisfying the statement of Proposition
                \ref{pro:classical limit} under the assumptions \eqref{eq:state}. Then $ \HHe $ defined in
                \eqref{eq:conv3} is self-adjoint on $\dom(\HHe)$, with form domain $ \dom(\sqrt{-\Delta + U_+}) $.
	\end{lemma}
	
	\begin{proof}
  		We are going to prove that $ \sum V_{\mu}(\xv_j) $ is form-bounded w.r.t. $ -\Delta$ with infinitesimally small bound. For any $z\in L^2 (\mathbb{R}^d )$ we split the potential into two pieces:
		\bml{
			\label{eq:Wz}
			\widecheck{ \lf(|\kv|^{\frac{1-d}{2}} z \ri)} (\xv)  = \frac{1}{(2\pi)^{\frac{d}{2}}} \int_{|\kv| \leq \varrho} \diff \kv \: e^{i\kv \cdot \xv}  |\kv|^{\frac{1-d}{2}} z(\kv) + \frac{1}{(2\pi)^{\frac{d}{2}}} \int_{|\kv| \geq \varrho} \diff \kv \: e^{i\kv \cdot \xv}  |\kv|^{\frac{1-d}{2}} z(\kv) \\
			=:W^{<}_{z}(\xv)+W_{z}^{>}(\xv)\;.
    		}
  		By Cauchy-Schwarz,
  		\beq
  			\label{eq:bounded part}
   			\sup_{\xv \in \R^d} \lf| W^{<}_{z}(\xv) \ri| \leq \frac{1}{2}\int_{\lf| \kv\ri|\leq \varrho}^{} \diff \kv \: \lf|\kv\ri|^{1-d} + \frac{1}{2(2\pi)^{d}} \lf\|  z  \ri\|_{L^2}^2 \leq C \lf( 1 + \lf\|  z  \ri\|_{L^2}^2 \ri),
  		\eeq
  		for any finite $ \varrho $. Therefore, since $ \lf\| z \ri\|_{L^2}^2 $ is integrable with respect to $\mathrm{d}\mu(z)$ by \eqref{eq:state} and Proposition \ref{pro:classical limit}, the potential associated to $ W^{<}_{z} $ is bounded:
  		\begin{equation}
    			\label{eq:19}
    			\sup_{\xv \in \mathbb{R}^d} \bigg| \int_{L^2(\R^d)}^{} \diff \mu(z) \: W^{<}_{z}(\xv) \bigg| \leq C \int_{L^2(\R^d)} \diff \mu(z) \: \lf(1 + \lf\|  z  \ri\|_{L^2}^2 \ri) \leq C < + \infty.
  		\end{equation}
  		In particular the quadratic form of
  		\bdm
  			2 \Re \sum_{j=1}^N  \int_{L^2(\R^d)}^{} \diff \mu(z) \: W^{<}_{z}(\xv_j)
		\edm
		is infinitesimally small w.r.t. to any positive operator
  
   		 In order to bound the second part, we rewrite
  		\begin{equation}
  			\label{eq:polaron trick}
    			W_{z}^{>}(\xv)= \bigg[ \int_{|\kv| \geq \varrho} \diff \kv \:\frac{\kv}{|\kv|^{\frac{d+3}{2}}} e^{-i\kv \cdot \xv}   z(\kv) + \mbox{c.c.} \:, \: i\nabla_{\xv} \bigg],
  		\end{equation}
  		so that, by Cauchy-Schwarz, for any $\psi \in \dom(\sqrt{-\Delta_{\xv}})$ and any $\alpha>0$,
  		\begin{equation}
    			\label{eq:20}
    			\lf| \meanlr{\psi}{W_{z}^{>}(\xv)}{\psi}_{L^2(\R^d)} \ri| \leq \alpha \meanlr{\psi}{-\Delta_{\xv}}{\psi} + \frac{1}{\alpha} \lf\| z \ri\|_{L^2(\R^d)}^2 \lf\| \psi \ri\|_{L^2(\R^d)}^2 \int_{|\kv| \geq \varrho} \diff \kv \:\frac{1}{|\kv|^{d+1}}. 
  		\end{equation}
 		Now since again the integrals of both $ 1 $ and $ \lf\| z  \ri\|_{L^2}^2$ against $\mathrm{d}\mu(z)$ are finite, it follows that
 		\bdm
  			2 \Re \sum_{j=1}^N  \int_{L^2(\R^d)}^{} \diff \mu(z) \: W^{>}_{z}(\xv_j)
		\edm
		{is Kato-infinitesimally small w.r.t. $ \sum_{j=1}^N - \Delta_j $.}
	\end{proof}
	
	\begin{remark}[Decay of $ V_\mu $]
	  	\label{rem:V decay}
	  	\mbox{}	\\
 		We point out that, if in addition to \eqref{eq:state}, $  \Psi_\varepsilon \in \dom (\sqrt{\mathrm{d}\Gamma(|\kv|^2)}) $, with uniform bound w.r.t. $\varepsilon$, then $ V_{\mu}$ is in fact continuous and vanishing at infinity. Indeed one can show that $ V_{\mu} $ is the Fourier (anti-)transform of an $ L^1 $ function: instead of using the trick \eqref{eq:polaron trick}, it suffices to apply Cauchy inequality twice, obtaining
 		\bmln{
 			\int_{L^2(\R^d)} \diff \mu(z) \bigg| \int_{|\kv| \geq \varrho} \diff \kv \: \frac{1}{|\kv|^{\frac{d-1}{2}}}  z(\kv)	+ \mbox{c.c.} \bigg| \\
 			\leq  \bigg[ \int_{L^2(\R^d)} \diff \mu(z) \bigg]^{1/2} \bigg[\int_{L^2(\R^d)} \diff \mu(z) \bigg| \int_{|\kv| \geq \varrho} \diff \kv \: \frac{1}{|\kv|^{\frac{d-1}{2}}}  z(\kv)	+ \mbox{c.c.} \bigg|^2 \bigg]^{1/2} \\
 			\leq 4\bigg[\int_{|\kv| \geq \varrho} \diff \kv \: \frac{1}{|\kv|^{ d+1 }} \bigg]^{1/2} \bigg[ \int_{L^2(\R^d)} \diff \mu(z) \: \lf\| |\kv| z \ri\|_{L^2}^2 \bigg]^{1/2} \leq C.
		}
		Hence the Fourier transform of that part of the potential belongs to $ L^1 $, but an identical property holds true for $ W^<_z $: exploiting again Cauchy inequality instead of \eqref{eq:bounded part} as above, one obtains immediately the result.
	\end{remark}

	The completion of the proof of Theorem \ref{teo:conv3} only requires to prove the convergence of $ \HH_{\eps} $ to $ \HHe $ in norm resolvent sense. Before attacking the proof we state however another useful technical result that will be used later.
	
	\begin{lemma}
		\label{lemma:gradient bound}
		\mbox{}	\\
		Let $ \zeta \in \R $ such that $ - \zeta \in \rho(\HH_{\eps}) \cap \rho(\HHe) $ belongs to resolvent sets of $ \HH_{\eps} $ and $ \HHe $ uniformly in $ \eps $. Then
		\beqn	
			\label{eq:gradient bound}
			\lf\| i \partial_{{j}} \lf(\HH_{\eps} + \zeta \ri)^{-1} \ri\| \leq C,		\qquad	
			\lf\| i \partial_{{j}} \lf(\HHe + \zeta \ri)^{-1} \ri\|  \leq  C,						
		\eeqn
		for any $ {j} = 1, \ldots, dN $.
	\end{lemma}
	
	\begin{proof}
		Let us prove the result for $ \HH_{\eps} $. The proof for $ \HHe $ is identical. We write 
		\bml{
			\label{eq:polaron proof 0}
			\lf\| i \partial_{{j}} \lf(\HH_{\eps} + \zeta \ri)^{-1} \ri\|^2 = \sup_{\lf\| \psi \ri\|_{2} = 1} \meanlr{\psi}{ \lf(\HH_{\eps} + \zeta \ri)^{-1} \lf( - \partial_{{j}}^2 \ri)  \lf(\HH_{\eps} + \zeta \ri)^{-1}}{\psi} \\
			\leq \sup_{\lf\| \psi \ri\|_{2} = 1} \meanlr{\psi}{ \lf(\HH_{\eps} + \zeta \ri)^{-1} \lf( -\Delta \ri)  \lf(\HH_{\eps} + \zeta \ri)^{-1}}{\psi}.
		}
		On the other relative boundedness of the quadratic form of $ V_{\eps,\Psi_{\eps}} $, positivity of $ U_+ $ and Kato smallness of $ U_{\ll} $ implies that $ - \Delta \leq - \Delta + U_+ $ and
		\bdm
			- \Delta + U_+ \leq \HH_{\eps} + a \lf( - \Delta + U_+ \ri) + b 
		\edm
		for any $ a < 1 $ and $ b $ finite, which implies that
		\bdm
			- \Delta \leq - \Delta + U_+ \leq C \lf( H_{\eps} + 1 \ri).
		\edm
		By replacing such an operator bound into \eqref{eq:polaron proof 0}, we thus get

		\bmln{
			\lf\| i \partial_{{j}} \lf(\HH_{\eps} + \zeta \ri)^{-1} \ri\|^2 \leq \sup_{\lf\| \psi \ri\|_{2} = 1} \meanlr{\psi}{ \lf(\HH_{\eps} + \zeta \ri)^{-1} C \lf( H_{\eps} + 1 \ri)  \lf(\HH_{\eps} + \zeta \ri)^{-1}}{\psi} \\
			\leq C \bigg[ \sup_{\lf\| \psi \ri\|_{2} = 1} \meanlr{\psi}{ \lf(\HH_{\eps} + \zeta\ri)^{-1}}{\psi} + (1 - \zeta) \lf\|  \lf(\HH_{\eps} + \zeta \ri)^{-1} \ri\|^2 \bigg] \leq C.
		}
	\end{proof}
	
	\begin{proof}[Proof of Theorem \ref{teo:conv3}]
		Thanks to Lemma \ref{lemma:sa hheps} and \ref{lemma:sa hhe}, it suffices to prove that $ \HH_{\eps} \to \HHe $ in norm resolvent sense. We first decompose both potentials as follows
		\beq
			V_{\eps,\Psi_{\eps}} = V_{\eps}^< + V_{\eps}^{>},	\qquad		V_{\mu} = V_{\mu}^< + V_{\mu}^>
		\eeq
		in order to distinguish low and high frequencies, as in the proof of Lemma \ref{lemma:sa hhe}. Indeed we set
		\beq
			V_{\mu}^{\#}(\xv) : = 2 \Re \int_{L^2(\R^d)} \diff \mu(z) \: W_z^{\#}(\xv),
		\eeq
		where the operators $ W_z^{\#} $, $ \# $ being either $ < $ or $ > $, are defined in \eqref{eq:Wz} and $ \varrho > 0  $ is a positive parameter. For $ V_{\eps,\Psi_{\eps}} $ we perform a similar decomposition at the level of the full quadratic form, {\it i.e.},
		\bml{
			V_{\eps,\Psi_{\eps}}(\xv) = \meanlr{\Psi_{\eps}}{A(g(\xv))}{\Psi_{\eps}}_{\fock} =   \int_{|\kv| \leq \varrho} \diff \kv \: e^{-i\kv \cdot \xv}  |\kv|^{\frac{1-d}{2}} \meanlr{\Psi_{\eps}}{a^{\dagger}_{\kv}}{\Psi_{\eps}}_{\fock} + \mbox{c.c.} \\
			+ \int_{|\kv| \geq \varrho} \diff \kv \: e^{-i\kv \cdot \xv}  |\kv|^{\frac{1-d}{2}} \meanlr{\Psi_{\eps}}{a^{\dagger}_{\kv} }{\Psi_{\eps}}_{\fock} + \mbox{c.c.} = : V_{\eps}^<(\xv) + V_{\eps}^{>}(\xv),
		}
		for the same $ \varrho $ as above. Now let $ \psi \in L^2(\R^d) $, and $ \zeta > 0  $ such that $ - \zeta \in \rho(\HH_{\eps}) \cap \rho(\HHe) $ uniformly in $ \eps $, {\it i.e.}, $ \dist(-\zeta, \sigma(\HH_{\eps})) > C > 0 $. Then the second resolvent identity yields
		\bml{
			\label{eq:polaron proof 1}
			\sup_{\lf\| \psi \ri\|_2 = 1} \lf\| \lf[ \lf(\HH_{\eps} +  \zeta \ri)^{-1} - \lf(\HHe + \zeta \ri)^{-1} \ri] \psi \ri\|_{L^2} \leq \sup_{\lf\| \psi \ri\|_2 = 1} \lf[
			\lf\| \lf(\HH_{\eps} +  \zeta \ri)^{-1} T_{\eps}^<  \lf(\HHe + \zeta \ri)^{-1}  \psi \ri\|_{L^2} \ri. \\
			\lf. +\lf\|   \lf(\HH_{\eps} + \zeta \ri)^{-1} T_{\eps}^> \lf(\HHe + \zeta \ri)^{-1} \psi \ri\|_{L^2} \ri],
		}
		where we have set
		\bdm
			T_{\eps}^{\#}(\xv_1, \ldots, \xv_N) : = \sum_{j = 1}^N \lf( V_{\mu}^{\#}(\xv_j) - V_{\eps}^{\#}(\xv_j) \ri)
		\edm
		for short. At this stage there is a technical subtlety not to be forgotten: the explicit expressions of the operators $ \HH_{\eps} $ and $ \HHe $ are a priori purely formal and make sense only when represented as quadratic forms. Therefore it is not obvious that the second resolvent identity could be used above and would lead to the the r.h.s. of \eqref{eq:polaron proof 1}. There is however a simple way out: by noting that
		\bdm	
			\lf\| \lf[ \lf(\HH_{\eps} +  \zeta \ri)^{-1} - \lf(\HHe + \zeta \ri)^{-1} \ri] \psi \ri\|_{L^2} = \sup_{\lf\| \phi \ri\|_2 \leq 1} {\lf| \braketr{\phi}{ \lf[ \lf(\HH_{\eps} +  \zeta \ri)^{-1} - \lf(\HHe + \zeta \ri)^{-1} \ri] \psi}_{L^2} \ri|},
		\edm
		one can express the vector norm in terms of a sesquilinear form, which in turn can be reduced to quadratic forms via polarization. At that level then one can use the operator expressions, being sure that the second resolvent identity makes sense and yields the r.h.s. of \eqref{eq:polaron proof 1}, because the form domains of $ \HH_{\eps} $ and $ \HHe $ are the same. 
		
		We now claim that both potentials $ V_{\eps}^<(\xv) $ and $ V_\mu^<(\xv) $ belongs to $ L^{\infty}(\R^d) $ uniformly in $ \eps $ and therefore
		\bdm
			\lf\| T_{\eps}^<(\xv_1, \ldots, \xv_N) \ri\|_{L^{\infty}(\R^{dN})} \leq C < +\infty.
		\edm
		This was already proven for $ V_{\mu} $ in \eqref{eq:19}, while, for $ V_{\eps}^{<} $, we act exactly as in \eqref{eq:bounded part} to estimate (recall that by assumption $ \Psi_{\eps} $ is normalized and \eqref{eq:state} holds)
		\bdm
			\lf| V_{\eps}^<(\xv) \ri| \leq \int_{|\kv| \leq \varrho} \diff \kv \:  \frac{1}{|\kv|^{d-1} } + \meanlr{\Psi_{\eps}}{\diff \Gamma(1)}{\Psi_{\eps}}_{\fock} \leq C(\varrho) < +\infty,
		\edm
		for any finite $ \varrho $. Therefore we can prove that the term involving $ T_{\eps}^< $ in \eqref{eq:polaron proof 1} tends to $ 0 $ as $ \eps \to 0 $ directly by a dominated convergence argument, using the pointwise convergence to $ 0 $ of $ T_{\eps}^<(\xv_1, \ldots, \xv_N) $, which is discussed in Corollary~\ref{cor:pointwise}:
		\bml{
			\sup_{\lf\| \psi \ri\|_2 = 1} \lf\| \lf(\HH_{\eps} +  \zeta \ri)^{-1} T_{\eps}^<  \lf(\HHe + \zeta \ri)^{-1}  \psi \ri\|_{L^2} \leq C \sup_{\lf\| \psi \ri\|_2 = 1} \lf\| T_{\eps}^<  \phi \ri\|_{L^2}\\
			 = \sup_{\lf\| \psi \ri\|_2 = 1} \sup_{\lf\| \xi \ri\|_2 \leq 1} \braketr{\xi}{ T_{\eps}^<  \phi}_{L^2} \xrightarrow[\eps \to 0]{} 0,
		}
		where we have set $ \phi : = \lf(\HHe + \zeta \ri)^{-1}  \psi \in L^2(\R^d) $.
		
		To complete the proof it remains then to consider the second term on the r.h.s. of \eqref{eq:polaron proof 1}. The idea is still to use the dominated convergence theorem but one has to exploit the trick \eqref{eq:polaron trick}: we rewrite
		\bml{
			\label{eq:polaron proof 2}
			\sup_{\lf\| \psi \ri\|_2 = 1} \lf\|   \lf(\HH_{\eps} +  \zeta \ri)^{-1} T_{\eps}^> \lf(\HHe + \zeta \ri)^{-1} \psi \ri\|_{L^2} \\
			= \sup_{\lf\| \psi \ri\|_2 = 1} \sum_{k = 1}^{dN} \lf\|   \lf(\HH_{\eps} +  \zeta \ri)^{-1} \lf[ \lf(\mathbf{S}_{\eps}(\xv_1, \ldots, \xv_N) \ri)_k,  i \partial_k \ri] \lf(\HHe + \zeta \ri)^{-1} \psi \ri\|_{L^2},
		}
		where $ \mathbf{S}_{\eps}(\xv_1, \ldots, \xv_N) $ is a vectorial multiplication operator given by
		\bml{
			\mathbf{S}_{\eps}(\xv_1, \ldots, \xv_N) : = \sum_{j =1}^N  \bigg[ \int_{|\kv| \geq \varrho} \diff \kv \: e^{-i\kv \cdot \xv_j}  \frac{\kv}{|\kv|^{\frac{d+3}{2}}} \meanlr{\Psi_{\eps}}{a^{\dagger}\lf(\kv\ri) }{\Psi_{\eps}}_{\fock} + \mbox{c.c.} \\
			- \int_{L^2(\R^d)} \diff \mu(z) \int_{|\kv| \geq \varrho} \diff \kv \: e^{-i\kv \cdot \xv}  \frac{\kv}{|\kv|^{\frac{d+3}{2}}}  z^{*}(\kv) - \mbox{c.c.} \bigg].
		}
		By Lemma \ref{lemma:gradient bound} and \eqref{eq:polaron proof 2},
		\bmln{
			\lf\|   \lf(\HH_{\eps} +  \zeta \ri)^{-1} T_{\eps}^> \lf(\HHe + \zeta \ri)^{-1} \psi \ri\|_{L^2} \leq C \sum_{k = 1}^{dN} \lf[ \lf\|   \lf(\mathbf{S}_{\eps}(\xv_1, \ldots, \xv_N) \ri)_k  \chi \ri\|_{L^2} \ri.	\\
			\lf. + \lf\|  \lf( \mathbf{S}_{\eps}(\xv_1, \ldots, \xv_N) \ri)_k \lf(\HHe + \zeta \ri)^{-1} \psi \ri\|_{L^2}\ri],
		}
		where $ \chi : = i \partial_k \lf(\HHe + \zeta \ri)^{-1} \psi \in L^2(\R^{dN}) $ with norm independent of $ \eps $ and we have used the boundedness of the resolvent $ \lf(\HH_{\eps} +  \zeta \ri)^{-1} $. Now for any $ k = 1, \ldots, dN $
		\beq
			 \lf( \mathbf{S}_{\eps}(\xv_1, \ldots, \xv_N) \ri)_k \xrightarrow[\eps \to 0]{\mathrm{pointwise}} 0,
		\eeq
		by a direct application of Corollary \ref{cor:pointwise}, since now $ \one_{[\varrho,\infty)}(|\kv|) \: |\kv|^{-\frac{d+1}{2}} \in L^2(\R^d) $, for any $ \varrho >  0 $. Moreover following the very same argument used above, one can show that
		\beq
			 \lf\| \lf( \mathbf{S}_{\eps}(\xv_1, \ldots, \xv_N) \ri)_k \ri\|_{L^{\infty}(\R^{dN})} \leq C < +\infty.
		\eeq
		Indeed, by \eqref{eq:state} and integrability of $ \lf\| z \ri\|_2^2 $ (see Proposition \ref{pro:classical limit})
		\bmln{
			\lf| \mathbf{S}_{\eps}(\xv_1, \ldots, \xv_N) \ri| \leq N  \bigg[ 2\int_{|\kv| > \varrho} \diff \kv \: \frac{1}{|\kv|^{d+1}} + \meanlr{\Psi_{\eps}}{\diff \Gamma(1)}{\Psi_{\eps}}_{\fock} 	+ \int_{L^2(\R^d)} \diff \mu(z) \: \lf\| z \ri\|_{L^2}^2  \bigg]	\\
			\leq C(\varrho) < +\infty,
		}
		for any $ \varrho > 0 $. In conclusion, by writing
		\bmln{
			\sup_{\lf\| \psi \ri\|_2 = 1} \lf\|   \lf(\mathbf{S}_{\eps}(\xv_1, \ldots, \xv_N) \ri)_k  \chi \ri\|_{L^2} + \lf\|  \lf( \mathbf{S}_{\eps}(\xv_1, \ldots, \xv_N) \ri)_k \lf(\HHe + \zeta \ri)^{-1} \psi \ri\|_{L^2} \\
			= \sup_{\lf\| \psi \ri\|_2 = 1} \sup_{\lf\| \xi \ri\|_2 \leq 1} \lf[ \braketr{\xi}{  \lf(\mathbf{S}_{\eps}(\xv_1, \ldots, \xv_N) \ri)_k  \chi}_{L^2} + \braketr{\xi}{\lf( \mathbf{S}_{\eps}(\xv_1, \ldots, \xv_N) \ri)_k \lf(\HHe + \zeta \ri)^{-1} \psi}_{L^2} \ri],
		}
		and using again dominated convergence, we get the result.
	\end{proof}

\subsection{Ground state energy: massive Nelson model}
\label{sec:gs nelson}

The setting in this Sect. is the one described in Sect. \ref{sec:ground state}. We start by considering the massive Nelson model and then comment on the adaptation required for the polaron.

We recall that the effective Hamiltonian for the particles in the quasi-classical limit is by Theorem \ref{teo:conv2}
\bmln{
  	\HHe(\mu) = \HH_0+2 \Re\sum_{j=1}^N\disp\int_{L^2(\mathbb{R}^d)}^{}\mathrm{d}\mu(z) \: \int_{\R^d} \diff \kv \: e^{i \kv \cdot \xv_j} z(\kv)  \lambda^{*}(\kv) \\
  	= \HH_0+2 (2\pi)^{d/2} \Re\sum_{j=1}^N\disp\int_{L^2(\mathbb{R}^d)}^{}\mathrm{d}\mu(z) \: \lf(\widecheck{ z\lambda^*}\ri) (\xv_j) 
}
where we have made explicit the dependence of $ \HHe $ on the classical measure $ \mu \in \M(L^2(\R^d)) $ provided by Proposition \ref{pro:classical limit}. The field energy in the limit $ \eps \to 0 $ becomes, under the assumptions \eqref{eq:state} (see again Proposition \ref{pro:classical limit}),
\bdm
  	c(\mu) = \int_{L^2(\R^d)}^{}\diff \mu(z) \: \lf\| \sqrt{\omega} z  \ri\|_{L^2}^2.
\edm
For further convenience we will denote the full energy of the system in the classical limit by
\beq
	\label{eq:kke}
  	\KKe(\mu) := \HHe(\mu)+c(\mu).
\eeq
Finally, we recall the minimization domain \eqref{eq:domain measure} for the measure $ \mu $:
\bdm
	\M_{\omega} : =\lf\{\mu\in \M \lf(L^2 (\mathbb{R}^d  )\ri) \: \Big| \: \mu\lf(L^2_{\omega} (\mathbb{R}^d )\ri) = 1, \: \mu\big|_{L^2_{\omega} (\mathbb{R}^d )}\text{ is Borel}, \: c(\mu) < \infty \ri\},
\edm
where
\bdm
	L^2_{\omega}(\R^d) : = \bigg\{ f \in L^2(\R^d) \: \bigg| \:  \int_{\R^d} \diff \kv \: \omega(\kv) \lf|f(\kv) \ri|^2 < \infty  \bigg\},
\edm
and $ \omega \geq 1 $. An important remark about measures in $ \M_{\omega}$ is that, although each $ \mu \in \M_{\omega} $ is a probability measure on $ L^2(\R^d) $, its support is totally concentrated in $ L^2_{\omega} (\mathbb{R}^d )  $, {\it i.e.}, the measure $ \mu $ vanishes outside $ L^2_{\omega} (\mathbb{R}^d ) $ and all the integrals involving $ \mu $ can be equivalently computed over $ L^2(\R^d) $ or $ L^2_{\omega} (\mathbb{R}^d )  $.

The first key result we prove is the boundedness from below of the infimum of the spectral bottom of $ \KKe(\mu) $ over $ \mu \in \M_{\omega} $:

	\begin{proposition}[Boundedness from below of $  \underline{\sigma}(\KKe(\mu)) $]
 		\label{pro:bd kke}
 		\mbox{}	\\
 		Let the assumptions \eqref{eq:V}, \eqref{eq:g2} and \eqref{eq:omega2} be satisfied, then
  		\beq
  			\label{eq:bd kke}
  			\inf_{\mu \in \M_{\omega}} \underline{\sigma}\lf( \KKe(\mu) \ri) \geq \underline{\sigma}(\HH_0) - N^2 \lf\| {\omega^{-1/2}} \lambda \ri\|_{L^2}^2 > -\infty.
		\eeq
	\end{proposition}

	\begin{proof}
		We first prove a simple but useful inequality: for any $ \mu \in \M_{\omega} $ and $ \delta > 0 $,
		\beq
			\label{eq:int bound}
			\bigg| \disp\int_{L^2(\mathbb{R}^d)}^{}\mathrm{d}\mu(z) \: \int_{\R^d} \diff \kv \: e^{i \kv \cdot \xv_j} z(\kv)  \lambda^{*}(\kv) \bigg| \leq \delta \disp\int_{L^2(\mathbb{R}^d)}^{}\mathrm{d}\mu(z) \: \lf\| \sqrt{\omega} z \ri\|_{L^2}^2 + \frac{1}{\delta} \lf\| \omega^{-1/2} \lambda \ri\|_{L^2}^2.
		\eeq
		Note that the r.h.s. is finite thanks to the assumptions on $ \lambda $ \eqref{eq:g2} and the hypothesis on $ \omega $. Here in particular it is important that we are considering the massive Nelson model to have $\omega\geq 1$. Using the above inequality we get
  		\bmln{
    		\underline{\sigma}\lf( \KKe(\mu) \ri) \geq \underline{\sigma}(\HH_0)+c(\mu)+2 (2\pi)^{d/2}\inf_{(\xv_1, \ldots, \xv_N) \in \mathbb{R}^{dN}} \Re\sum_{j=1}^N\disp\int_{L^2(\mathbb{R}^d)}^{}\mathrm{d}\mu(z) \: \lf(\widecheck{ z\lambda^*}\ri) (\xv_j)	\\ 	
    		\geq \underline{\sigma}(\HH_0)+c(\mu) - N \delta \disp\int_{L^2(\mathbb{R}^d)}^{}\mathrm{d}\mu(z) \: \lf\| \sqrt{\omega} z \ri\|_{L^2}^2 - \frac{N}{\delta} \lf\| \omega^{-1/2} \lambda \ri\|_{L^2}^2	\\
    		 \geq \underline{\sigma}(\HH_0) - N^2 \lf\| \omega^{-1/2} \lambda \ri\|_{L^2}^2 > -\infty,
  		}
  		where we have taken $ \delta = \frac{1}{N} $ in the last step. Since the r.h.s. is independent of $ \mu $, we conclude that \eqref{eq:bd kke} holds true.
	\end{proof}

An important consequence of Proposition \ref{pro:bd kke} is that the infimum can be taken over measures with finite
support, {\it i.e.}, finite linear combinations of Dirac masses. We set
\bml{
	\label{eq:mfin}
	\Mfin : = \bigg\{ \mu \in \M_{\omega} \Big| \: \exists I \subset \N \mbox{ finite}, \lf\{ \alpha_i \ri\}_{i\in I} \in \R^+, \sum_{i \in I} \alpha_i = 1, \lf\{ z_i \ri\}_{i \in I} \in L^2_{\omega}(\R^d),	\\
	\mbox{s.t. } \mu = \sum_{i \in I} \alpha_i \delta(z - z_i) \bigg\}.
}

	\begin{lemma}
  		\label{lemma:finite measures}
  		\mbox{}	\\
  		Under the assumptions of Proposition \ref{pro:bd kke},
  		\beq
  			\label{eq:finite measures}
  			\inf_{\mu \in \M_{\omega}} \underline{\sigma}\lf( \KKe(\mu) \ri) = \inf_{\mu \in \Mfin} \underline{\sigma}\lf( \KKe(\mu) \ri).
		\eeq
	\end{lemma}

	\begin{proof}
  		Since for any $ z_0 \in  L^2_{\omega}(\mathbb{R}^d  )$, $ \delta(z - z_0) \in \M_{\omega} $, it immediately follows that
  		\beq
  			\label{eq:ub finite}
    			\inf_{\mu \in \M_{\omega}} \underline{\sigma}\lf( \KKe(\mu) \ri) \leq \inf_{\mu \in \Mfin} \underline{\sigma}\lf( \KKe(\mu) \ri).
  		\eeq
  
  		Now let $ \M_{\omega} $ be endowed with the $2$-Wasserstein distance. Since $ L^2_{\omega} $ is separable and complete, then $ \M_{\omega} \simeq \M_2(L^2(\R^d, \omega \: \diff \kv)) $, the space of measures on $ L^2(\R^d, \omega \: \diff \kv) $ with finite 2-moments, is separable and complete. In addition, $ \Mfin $ is dense in $ \M_2(L^2(\R^d, \omega \: \diff \kv)) $ (see, {\it e.g.}, \cite[Theorem 6.18]{MR2459454}). By the isomorphism $ \M_{\omega} \simeq \M_2(L^2(\R^d, \omega \: \diff \kv)) $, it follows that the atomic measures are dense in $ \Mfin $ w.r.t. the topology induced by $\M_2(L^2(\R^d, \omega \: \diff \kv))$.

  		Now, let us assume that for any $ \mu \in \M_{\omega} $ and any sequence $ \lf\{ \mu_n \ri\}_{n \in \N} $ converging to $ \mu $ in $ \M_{2} $, one has
  		\beq
  		\label{eq:8}
    			\exists \lf\{ \mu_{n_k} \ri\}_{k\in \mathbb{N}} \text{  s.t. } \lf| \underline{\sigma}\lf( \KKe(\mu) \ri) - \underline{\sigma}\lf( \KKe\lf(\mu_{n_k}\ri) \ri) \ri| \xrightarrow[k\to \infty]{} 0.
  		\eeq
  		By definition of infimum, there exists some $ \nu \in \M_{\omega} $, such that, for any $ \delta > 0 $,
 		\bdm
    			\underline{\sigma} \lf( \KKe(\nu) \ri) < \inf_{\mu \in \M_{\omega}} \underline{\sigma}\lf( \KKe(\mu) \ri) + \tx\frac{1}{2}\delta.
		\edm
  		Thanks to the density of $ \Mfin $ in $ \M_{\omega} $, there must exist a sequence $ \lf\{ \nu_n \ri\}_{n \in \N} \in \Mfin $, such that 
  		\bdm
  			\nu_n \xrightarrow[n \to \infty]{\M_2} \nu.
		\edm
		Then \eqref{eq:8} implies that there exists at least one subsequence $ \lf\{ \nu_{n_k} \ri\}_{k\in \mathbb{N}}$ and a $\bar{k}\in \mathbb{N}$ such that for all $k\geq \bar{k}$:
  		\bdm
    			\lf| \underline{\sigma}\lf( \KKe(\nu) \ri) - \underline{\sigma}\lf( \KKe\lf(\nu_{n_k}\ri) \ri) \ri|< \tx\frac{1}{2}\delta,
 		\edm
 		which implies	
  		\bdm
  			\inf_{\mu\in \M_{\omega}} \underline{\sigma} \lf( \KKe(\mu) \ri) > \underline{\sigma}\lf( \KKe(\nu)\ri) - \tx\frac{1}{2}\delta
    > \disp\underline{\sigma}\lf( \KKe(\nu_{n_k})\ri) - \delta > \disp\inf_{\mu \in \Mfin} \underline{\sigma}\lf( \KKe(\mu)\ri) - \delta.
  		\edm
  		Since $\delta>0$ is arbitrary, it follows that
  		\beq
  			\label{eq:lb finite}
    			\inf_{\mu \in \M_{\omega}} \underline{\sigma}\lf( \KKe(\mu) \ri) \geq \inf_{\mu \in \Mfin} \underline{\sigma}\lf( \KKe(\mu) \ri),
  		\eeq
  		which together with \eqref{eq:ub finite} yields the result.

  		The only thing that has yet to be proved is the statement \eqref{eq:8}. If $\mu_n \longrightarrow \mu$ in $ \M_2 $, then $ \lf\{ \mu_n \ri\}_{n\in \mathbb{N}}$ is Cauchy and therefore tight (and precompact by Prokhorov's Theorem) in $ \M(L^2_{\omega}) $ (see \cite[Lemma 6.14]{MR2459454}). Therefore there exists a subsequence $ \lf\{ \mu_{n_k} \ri\}_{k\in \mathbb{N}}$, such that $\mu_{n_k} \longrightarrow \mu$ as $ k \to \infty $ in both $ \M $ and $ \M_2 $ topologies. Then the following convergence holds:
 		\bmln{	
      		\lf\| \KKe(\mu) - \KKe(\mu_{n_k}) \ri\| \leq N \sup_{\xv \in \R^d} \lf| V_\mu(\xv) - V_{\mu_{n_k}}(\xv) \ri| + \lf| c(\mu) - c(\mu_{n_k}) \ri| \\
      		\leq N \lf\| \omega^{-1/2} \lambda  \ri\|_{L^2}^2 \bigg| \int_{L^2_{\omega}(\R^d)}^{}  \mathrm{d}(\mu-\mu_{n_k})(z) \bigg| \\
      		+(N+1) \bigg| \int_{L^2_{\omega}(\R^d)}^{} \mathrm{d}(\mu-\mu_{n_k})(z)  \: \lf\| \sqrt{\omega} z  \ri\|_{L^2}^2 \bigg| \xrightarrow[k \to \infty]{} 0.
    		}	
  		Finally, by \cite[Theorem 4.10]{Kabo}, the distance between the spectra of $ \KKe(\mu) $ and $ \KKe(\mu_{n_k}) $ converges to zero and the same is true for the ground state energy and thus \eqref{eq:8} is proven.
	\end{proof}

We can now complete the proof of our main result on the ground state energy convergence:

	\begin{proof}[Proof of Theorem \ref{teo:gse}]
		First of all, we observe that assumptions \eqref{eq:state} and Proposition \ref{pro:classical limit} guarantee that the classical measure $ \mu  $ is actually supported on $ L^2_{\omega}(\R^d) $ and $ c(\mu) < + \infty $ (see Corollary \ref{cor:field energy}). Moreover since it is a probability measure, it must be $ \mu(L^2_{\omega}(\R^d)) = 1 $.  Finally, the fact that each measure obtained in the quasi-classical limit is Borel when restricted to $ L^2_{\omega}(\R^d) $ can be proven by adapting \cite[ Proposition 3.11]{2011arXiv1111.5918A}. We have therefore justified the restriction to $ \M_{\omega} $ in the minimization of the energy form. 
		
		Next we recall that $ C_0^{\infty}\lf(\mathbb{R}^{dN}\ri) \otimes \dom \lf(\mathrm{d}\Gamma(\omega)\ri)$ is a core for $H$, and $ C_0^{\infty}\lf(\mathbb{R}^{dN}\ri)$ is a core for  $ \HHe(\mu)$, for any $\mu \in \M(L^2(\R^d)) $, which is a consequence of the fact that $ \HHe(\mu) $ is self-adjoint on $ \dom(\HH_0) $ (Theorem \ref{teo:conv2}) and $ \HH_0 $ is essentially self-adjoint on $ C_0^{\infty}\lf(\R^{dN}\ri) $. 
		
		In addition, $\underline{\sigma}(H)$ is uniformly bounded from below w.r.t. $\varepsilon $ (see Proposition \ref{pro:sa1} but also Proposition \ref{pro:kato-rellich}). Now let $\varphi\in C_0^{\infty}\lf(\mathbb{R}^{dN}\ri)$, by Lemma \ref{lemma:finite measures} it suffices to compute the energy for measures $ \mu \in  \Mfin $: let then $ I $ be a finite subset of $\N $, $ \lf\{ \alpha_i \ri\}_{i\in I} \in \R^+ $ such that $ \sum_{i \in I} \alpha_i = 1 $, $ \lf\{ z_i \ri\}_{i \in I} \in L^2_{\omega} $,	 so that we can express $ \mu $ as a convex combination of Dirac masses, {\it i.e.},
		\bdm
			\mu = \sum_{i \in I} \alpha_i \delta(z - z_i).
		\edm
		Then, by linearity (recall the definition of coherent states \eqref{eq:coherent} and Proposition \ref{pro:coherent states}),
 		\bmln{
      		\underline{\sigma}(H) \leq \inf_{\mu \in \Mfin} \inf_{\varphi \in C^{\infty}_0(\R^{dN})} \sum_{i \in I} \alpha_i \meanlr{\varphi \otimes \Xi(z_j)}{H}{\varphi \otimes \Xi(z_j)}_{L^2\otimes \Gamma_{\mathrm{sym}}}	\\
      		= \inf_{\mu \in \Mfin} \inf_{\varphi \in C^{\infty}_0(\R^{dN})} \meanlr{\varphi}{\mathcal{K}_{\mathrm{eff}}(\mu)}{\varphi}_{L^2} =  \inf_{\mu \in \Mfin} \underline{\sigma}\lf(\mathcal{K}_{\mathrm{eff}}(\mu)\ri) = \inf_{\mu \in \M_{\omega}} \underline{\sigma}\lf(\mathcal{K}_{\mathrm{eff}}(\mu)\ri),
    		}
  		where the last equality holds by Lemma \ref{lemma:finite measures}.
  		
 		It remains to show that
  		\begin{equation*}
    			\liminf_{\varepsilon\to 0}\underline{\sigma}(H) \geq  \inf_{\mu \in \M_{\omega}} \underline{\sigma}\bigl(\mathcal{K}_{\mathrm{eff}}(\mu)\bigr)\; .
  		\end{equation*}
  		Let $\Pi_{\varepsilon,\delta}\in C_0^{\infty}(\mathbb{R}^{dN})\otimes \mathscr{D}(\mathrm{d}\Gamma(\omega))$, $\delta>0$, be a vector that satisfies
  		\begin{equation*}
    			\langle \Pi_{\varepsilon,\delta}\lvert H \rvert \Pi_{\varepsilon,\delta} \rangle_{}< \underline{\sigma}(H)+\delta\; .
  		\end{equation*}
		The simple operator bound
  		\begin{multline*}
   			\meanlr{\Psi}{H}{\Psi}_{} \geq \tx\frac{1}{2} \meanlr{\Psi}{\diff \Gamma(\omega)}{\Psi}_{} - \lf( 2 \lf\| \omega^{-1/2} \lambda \ri\|_{L^2}^2 + M \ri) \lf\| \Psi \ri\|^2_{}	\\
    \geq \tx\frac{1}{2} \meanlr{\Psi}{\diff \Gamma(\omega)}{\Psi}_{} - C \lf\| \Psi \ri\|^2_{},
  		\end{multline*}
  		for some $ M, C < +\infty $, which follows from boundedness from below of $ \HH_0 $, and (3.49) (with $ \delta = \frac{1}{2} $) together with (3.55), imply the
  that the expectation value of $ \diff \Gamma(\omega) $ on $\Pi_{\varepsilon,\delta}$ is uniformly bounded in $\eps$. Since
  $\mathrm{d}\Gamma(1)\leq \mathrm{d}\Gamma(\omega)$ for a massive field, both conditions in (\textbf{A}4) are therefore fulfilled by
  $ \Pi_{\eps, \delta}$. 

  		Now, the vectors of the form $\psi\otimes \Xi(f)$, $\psi\in C_0^{\infty}(\mathbb{R}^{dN})$ and $f\in L^2_{\omega}(\mathbb{R}^d)$, are total in
  $L^2(\mathbb{R}^{dN})\otimes \Gamma_{\mathrm{sym}}(L^2(\mathbb{R}^d))$, and belong to
  $C_0^{\infty}(\mathbb{R}^{dN})\otimes \mathscr{D}(\mathrm{d}\Gamma(\omega))$. Let us recall that $\Xi(f)$ is the squeezed coherent state defined in
  (2.33). Hence it is possible to choose, for any $\delta>0$, the vector $\Pi_{\varepsilon,\delta}$ of the form
  		\begin{equation*}
   			\Pi_{\varepsilon,\delta}=\sum_{i=1}^{M(\delta)}\lambda_{i,\delta}(\varepsilon)\psi_{i,\delta}\otimes \Xi\bigl(z_{i,\delta}\bigr)\; ,
  		\end{equation*}
  		where $z_{i,\delta}\neq z_{k,\delta}$ for all $i\neq k$, each $\psi_{i,\delta},\psi_{k,\delta}$ is normalized, and the $\lambda_{i,\delta}(\varepsilon)$ satisfy
 		\begin{equation}
    			\label{eq:2}
    			\sum_{i=1}^{M(\delta)}\lvert \lambda_{i,\delta}(\varepsilon)  \rvert_{}^2+\sum_{i<k}^{}2\Re \; \bar{\lambda}_{i,\delta}(\varepsilon)\lambda_{k,\delta}(\varepsilon)\bar{\psi}_{i,\delta}\psi_{k,\delta} e^{-\frac{i}{\varepsilon}\Im\langle z_{i,\delta} \lvert z_{k,\delta} \rangle_2-\frac{1}{2\varepsilon}\lVert z_{i,\delta}-z_{k,\delta}  \rVert_2^{2}}=1\; .
  		\end{equation}
  		Now,
  		\begin{equation*}
    			\Bigl\lvert 2\Re \bar{\lambda}_{i,\delta}(\varepsilon)\lambda_{k,\delta}(\varepsilon)\bar{\psi}_{i,\delta}\psi_{k,\delta}\Bigr\rvert \leq \lvert \lambda_{i,\delta}  \rvert_{}^2+\lvert \lambda_{k,\delta}  \rvert_{}^2\; ,
  		\end{equation*}
  		so \eqref{eq:2} yields
  		\begin{equation*}
    			\sum_{i=1}^{M(\delta)}(1-C_{i,\delta}(\varepsilon))\lvert \lambda_{i,\delta}(\varepsilon)  \rvert_{}^2 \leq 1\; ,
  		\end{equation*}
  		where
  		$0\leq \lvert C_{i,\delta}(\varepsilon)\rvert <2M(\delta)\max_{i<k\in M(\delta)} e^{-\frac{1}{2\varepsilon}\lVert z_{i,\delta}-z_{k,\delta}
    \rVert_2^{2}}$, and $C_{i,\delta}(\varepsilon)\to 0$ as $\varepsilon\to 0$. Therefore it follows that each
  $(\lambda_{i,\delta}(\varepsilon))_{\varepsilon\in (0,1)}$ is uniformly bounded for $\varepsilon$ small enough, \emph{e.g.} such that $C_{i,\delta}(\varepsilon)\leq 1/2$ for any $i$.

 		Given $\Pi_{\varepsilon,\delta}$ of this form, the corresponding expectation of $H$ can be thus explicitly computed quite easily, and takes the form
		\bml{
			\label{eq:1}
      			\langle \Pi_{\varepsilon,\delta}\lvert H \rvert \Pi_{\varepsilon,\delta} \rangle= \sum_{i,k=1}^{M(\delta)} \bar{\lambda}_{i,\delta}(\varepsilon)\lambda_{k,\delta}(\varepsilon)\langle z_{i,\delta}  \lvert \omega \rvert z_{k,\delta} \rangle_2 e^{-\frac{i}{\varepsilon}\Im \langle z_{i,\delta} \lvert z_{k,\delta} \rangle_2-\frac{1}{2\varepsilon}\lVert z_{i,\delta}-z_{j,\delta}  \rVert_2^2}\\+\sum_{i,k=1}^{M(\delta)}\bar{\lambda}_{i,\delta}(\varepsilon)\lambda_{j,\delta}(\varepsilon)\Bigl(\langle \psi_{i,\delta}  \lvert\mathcal{H}_0\rvert \psi_{k,\delta}\rangle_2 +\Bigl\langle\psi_{i,\delta}\Bigl\lvert\langle g(\mathbf{x})  \lvert z_{k,\delta} \rangle_2\\+\langle z_{i,\delta}  \lvert g(\mathbf{x}) \rangle_2\Bigr\rvert \psi_{k,\delta}\Bigl\rangle_2\Bigr) e^{-\frac{i}{\varepsilon}\Im \langle z_{i,\delta} \lvert z_{k,\delta} \rangle_2-\frac{1}{2\varepsilon}\lVert z_{i,\delta}-z_{j,\delta}  \rVert_2^2}\; .
    		}
  		Now let $(\lambda_{i,\delta})_{i=1}^{M(\delta)}\subset \mathbb{C}$ be cluster points of each
  $(\lambda_{i,\delta}(\varepsilon))_{\varepsilon\in (0,1)}$ corresponding to a common subsequence, satisfying
  		\begin{equation*}
    			\sum_{i=1}^{M(\delta)}\lvert \lambda_{i,\delta}  \rvert_{}^2=1\; .
 	 	\end{equation*}
 	 	Then the corresponding cluster point of~\eqref{eq:1} has the form
  		\bmln{
    		  		\sum_{i=1}^{M(\delta)} \lvert \lambda_{i,\delta} \rvert_{}^2 \Bigl(\lVert \omega^{1/2} z_{i,\delta} \rVert_2^2+\lvert \lambda_{i,\delta}  \rvert_{}^2\langle \psi_{i,\delta}  \lvert\mathcal{H}_0\rvert \psi_{i,\delta}\rangle_2 +\Bigl\langle\psi_{i,\delta}\Bigl\lvert 2\Re\langle z_{i,\delta}  \lvert g(\mathbf{x}) \rangle_2\Bigr\rvert \psi_{i,\delta}\Bigl\rangle_2 \Bigr)\\=\sum_{i=1}^{M(\delta)}\lvert \lambda_{i,\delta}  \rvert_{}^2 \bigl\langle \psi_{i,\delta}\bigl\lvert \mathcal{K}_{\mathrm{eff}}\bigl(\delta(\cdot - z_{i,\delta})\bigr) \bigr\rvert \psi_{i,\delta}\bigr\rangle_2\; .
		}
	  	Therefore, defining
	  	\bmln{
	     			\mathscr{D}_{\mathrm{gs}}(\delta)=\biggl\{(\lambda_{i,\delta})_{i=1}^{M(\delta)}\subset \mathbb{C}\; ,\; \sum_{i=1}^{M(\delta)}\lvert \lambda_{i,\delta}  \rvert_{}^2=1\; ; \;\{z_{i,\delta}\}_{i=1}^{M(\delta)}\subset L^2_{\omega}(\mathbb{R}^d)\; ;\\ \{\psi_{i,\delta}\}_{i=1}^{M(\delta)}\subset \{\lVert \psi  \rVert_2^{}=1\}\cap L^2(\mathbb{R}^{dN})\biggr\}
	    	}
	  	the $\liminf_{\varepsilon\to 0}$ of $\langle \Pi_{\varepsilon,\delta}\lvert H \rvert \Pi_{\varepsilon,\delta} \rangle$ takes the form
	  	\bmln{
      			\liminf_{\varepsilon\to 0}\;\langle \Pi_{\varepsilon,\delta}\lvert H \rvert \Pi_{\varepsilon,\delta} \rangle=\inf_{\mathscr{D}_{\mathrm{gs}(\delta)}} \; \sum_{i=1}^{M(\delta)}\lvert \lambda_{i,\delta}  \rvert_{}^2 \bigl\langle \psi_{i,\delta}\bigl\lvert \mathcal{K}_{\mathrm{eff}}\bigl(\delta(\cdot - z_{i,\delta})\bigr) \bigr\rvert \psi_{i,\delta}\bigr\rangle_2\\\geq \inf_{\substack{\lVert \lambda_{i,\delta}  \rVert_{\ell^{2}}^{}=1,\\\{z_{i,\delta}\}\subset L^2_{\omega}(\mathbb{R}^d)}}\; \sum_{i=1}^M\lvert \lambda_{i,\delta}  \rvert_{}^2\underline{\sigma}\Bigl(\mathcal{K}_{\mathrm{eff}}\bigl(\delta(\cdot -z_{i,\delta})\bigr)\Bigr)\\\geq \inf_{\mu\in \M_{\mathrm{fin}}}\underline{\sigma}(\mathcal{K}_{\mathrm{eff}}(\mu))=\inf_{\mu\in \M_{\omega}}\underline{\sigma}(\mathcal{K}_{\mathrm{eff}}(\mu))\; .
    		}
 	 	It then follows that
  		\begin{equation*}
    			\inf_{\mu\in \M_{\omega}}\underline{\sigma}(\mathcal{K}_{\mathrm{eff}}(\mu))\leq \underline{\sigma}(H)+\delta\; ,
  		\end{equation*}
  		and that concludes the proof.
	\end{proof}

\subsection{Ground state energy: polaron}
\label{sec:gs polaron}

The proof of Theorem \ref{teo:gse} for the polaron model goes along the same lines as
for the massive Nelson model, so we will focus mostly on the points where the two proofs differ. 

As we are going to see, the technical differences are related to the fact that most quantities involved in the proof have to be expressed as quadratic forms, since they do not make sense as operators. We thus set
\beq
	\Q_{\mu}[\psi] : = \meanlr{\psi}{\HHe(\mu)}{\psi}_{L^2(\R^{dN})},		\qquad		\T_{\mu}[\psi] : = \Q_{\mu}[\psi] + c(\mu) \lf\| \psi \ri\|_{L^2(\R^{dN})}^2,
\eeq
where, as in the previous Sect.,
\bdm
	c(\mu) = \int_{L^2(\R^d)} \diff \mu(z) \: \lf\| z \ri\|_{L^2(\R^d)}^2.
\edm
Note that the minimization domain for the measure $ \mu $ becomes
\beq
  	\dom_1 = \M_2\lf(L^2(\R^d)\ri),
\eeq
which is a separable and complete metric space once endowed with the 2-Wasserstein distance.	

The analogue of Proposition \ref{pro:bd kke} is formulated in next Proposition, which is a direct consequence of the KLMN Theorem.

	\begin{proposition}[Boundedness from below of $ \mathcal{T}_{\mu} $]
 		\label{pro:bd tmu}
 		\mbox{}	\\
 		Let the assumptions \eqref{eq:V}, \eqref{eq:omega1} and \eqref{eq:g3} be satisfied, then
  		\beq
  			\label{eq:bd tmu}
  			\inf_{\mu \in \M_{\omega}} \inf_{\lf\| \psi \ri\|_2 = 1} \T_{\mu}[\psi] \geq - C.
		\eeq
	\end{proposition}
	
	\begin{proof}
          The result follows from KLMN Theorem (see the proof of Lemma~\ref{lemma:sa hhe} and its
          notation). If we set \bdm C_<(\varrho) := \int_{\lf| \kv \ri| \leq \varrho}^{} \diff \kv \: \lf| \kv
          \ri|^{1-d}, \qquad C_>(\varrho) := \int_{\lf| \kv \ri| \geq \varrho}^{} \diff \kv \: \lf| \kv \ri|^{1-d},
          \edm and take $\alpha = \frac{1}{4N}$ in \eqref{eq:20} and combine it with \eqref{eq:19}, we get the
          following bound for any $\varrho\geq 0$ and any $\mu\in \dom_1 $:
  		\begin{equation}
  			\label{eq:form bound below}
    			\frac{\Q_{\mu}[\psi]}{\lf\| \psi \ri\|_{L^2}^2} > - \lf(8N^2C_>(\varrho)+N\alpha_1 \ri)\int_{L^2}^{} \mathrm{d}\mu(z) \: \lf\| z \ri\|_{L^2}^2 - \frac{N}{\alpha_1}C_<(\varrho).
  		\end{equation}
  		Now choosing $ \varrho, \alpha_1 $ in such a way that $ 8N^2C_>(\varrho)+N\alpha_1 = 1$, we obtain
 		\begin{equation*}
    			\frac{\Q_{\mu}[\psi]}{\lf\| \psi \ri\|_{L^2}^2} > -\int_{L^2}^{} \mathrm{d}\mu(z) \: \lf\| z \ri\|_{L^2}^2 - C,
  		\end{equation*}
  		where the constant $ C < +\infty $ on the r.h.s. is independent of $ \mu $.	Using the explicit expression of $ c(\mu) $, we immediately get
  		\beq
  			\T_{\mu}[\psi] \geq - C \lf\| \psi \ri\|_2^2.
		\eeq
	\end{proof}

	The only other argument to be adapted is the proof of Lemma \ref{lemma:finite measures}. In fact most of the proof does not need to be changed at all, whereas \eqref{eq:8} requires a totally different approach. We thus state the result as a separate lemma.

	\begin{lemma}
  		\label{lemma:8 polaron}
  		\mbox{}	\\
  		Let $ \lf\{ \mu_n \ri\}_{n \in \N} $ be a sequence of measures in $ \M_2(L^2(\R^d)) $ such that 
  		\bdm
  			\mu_n \xrightarrow[n \to \infty]{\M_2} \mu.
		\edm
		Then there exists a subsequence $ \lf\{ \mu_{n_k} \ri\}_{k\in \N} $, such that
 		 \beq
 		 	\label{eq:8 polaron}
    			\lf| \inf_{\lf\| \psi \ri\|_2 = 1} \T_{\mu}[\psi] - \inf_{\lf\| \psi \ri\|_2 = 1} \T_{\mu_{n_k}}[\psi] \ri| \xrightarrow[k \to \infty]{} 0.
  		\eeq
	\end{lemma}

	\begin{proof}
  		Since $ \lf\{ \mu_n \ri\}_{n\in \mathbb{N}}$ is Cauchy and therefore tight in $ \M\lf( L^2 (\mathbb{R}^d ) \ri)$, there exists a subsequence
  $ \lf\{ \mu_{n_k} \ri\}_{k\in \mathbb{N}}$, such that
  		\bdm
  			\mu_{n_k} \xrightarrow[k \to \infty]{\M} \mu, 		\qquad \mu_{n_k} \xrightarrow[k \to \infty]{\M_2} \mu.
		\edm
  
  		Now let $\delta>0$ be fixed and let $ \lf\{ \psi_j \ri\}_{j \in \mathbb{N}} \subset H^1(\mathbb{R}^d)$ be a minimizing sequence
  for $ \T_{\mu}[\psi] $. The existence of such a minimizing sequence for $ \T_{\mu}[\psi] $, $ \mu \in \M_1 $, is guaranteed by Proposition \ref{pro:bd tmu}. Then
 		\bdm
 			 \inf_{\lf\| \psi \ri\|_2 = 1} \T_{\mu_n}[\psi] \leq \Q_{\mu_n}\lf[\psi_j\ri] + c({\mu_n}) = \Q_{\mu}\lf[\psi_j\ri] + c(\mu) + \Q_{\mu_n-\mu}\lf[\psi_j\ri] + c(\mu_n - \mu),	
		\edm
  		Therefore there exists some $\bar{j} \geq 0$, such that, for any $ j \geq \bar{j} $,
  		\begin{equation*}
    			\inf_{\lf\| \psi \ri\|_2 = 1} \T_{\mu_n}[\psi] < \inf_{\lf\| \psi \ri\|_2 = 1} \T_{\mu}[\psi] +\Q_{\mu_n-\mu}\lf[\psi_j\ri] + c(\mu_n - \mu) + \tx\frac{1}{4} \delta.
  		\end{equation*}
  		In addition the convergence of $ \mu_{n_k} $ in $ \M_2 $ guarantees that there exists some $ \bar{n}_1 \in \N$ such that, for any $ n_k \geq \bar{n}_1 $, $ \lf| c(\mu - \mu_{n_k}) \ri| \leq \frac{\delta}{2} $ and therefore
  		\begin{equation*}
    			\inf_{\lf\| \psi \ri\|_2 = 1} \T_{\mu_n}[\psi]  < \inf_{\lf\| \psi \ri\|_2 = 1} \T_{\mu}[\psi] + \Q_{\mu_n-\mu}\lf[\psi_j\ri] + \tx\frac{1}{2} \delta,
  		\end{equation*}
  		for any $ n_k \geq \bar{n}_1 $ and $ j \geq \bar{j} $.
  		On the other hand, following the same lines leading to \eqref{eq:19} and \eqref{eq:20}, one can bound
	 	\begin{equation}
    			\label{eq:21}
    			\lf| \Q_{\mu_n-\mu}\lf[\psi\ri] \ri| \leq C \lf\| \psi \ri\|_{H^1(\R^{dN})} \bigg[ \int_{L^2(\R^d)}^{}\diff\lf|\mu - \mu_{n_k}\ri| \: \lf( 1+ \lf\| z \ri\|_{L^2}^2 \ri) \bigg]^{1/2}\leq  \delta^{\prime} \lf\| \psi \ri\|_{H^1(\R^{dN})},
  		\end{equation}
  		for any $ \delta^{\prime} > 0 $ and big enough $n_k$, thanks to the convergence of $ \mu_{n_k} $ to
                $ \mu $ in $ \M $ and $ \M_2 $. Indeed, this guarantees the existence of a $ \bar{n}_2 \in \N $ such that \bdm
                \int_{L^2(\R^d)}^{}\diff\lf|\mu - \mu_{n_k}\ri| \: \lf( 1+ \lf\| z \ri\|_{L^2}^2 \ri) \leq
                {\delta^{\prime}}^2, \edm for any $ n_k \geq \bar{n}_2 $ and $ \delta^{\prime} $ arbitrary. Hence, if
                $n_k \geq \bar{n} : = \max\{\bar{n}_1, \bar{n}_{2} \}$, \beq
  			\label{eq:formub}
    			\inf_{\lf\| \psi \ri\|_2 = 1} \T_{\mu_n}[\psi]  < \inf_{\lf\| \psi \ri\|_2 = 1} \T_{\mu}[\psi] +  2 \delta^{\prime} \lf\| \psi_j \ri\|_{H^1(\R^{dN})} + \tx\frac{1}{2} \delta 
    			\leq \inf_{\lf\| \psi \ri\|_2 = 1} \T_{\mu}[\psi] +  \delta,
  		\eeq
  		where we have taken $ \delta^{\prime} = \frac{1}{4} \lf\| \psi_j \ri\|^{-1}_{H^1(\R^{dN})} \delta$.	Since $ j \geq \bar{j} $, which does depend on $ \delta $, in order to show that such a choice of $ \delta^{\prime} $ is possible, we have to ensure that $ \lf\| \psi_j \ri\|_{H^1(\R^{dN})} $ is uniformly bounded. This is, however, a direct consequence of the following inequality
  		\beq
  			\label{eq:formlb}
  			\T_{\mu}[\psi] \geq (1 - \alpha) \lf\| \psi \ri\|^2_{H^1(\R^{dN})} - C_{\alpha} \lf\| \psi \ri\|_{L^2(\R^{dN})}^2,
		\eeq
		for any $ \alpha > 0 $, where the coefficient $ C_{\alpha} $ is finite for any $ \alpha > 0 $. Such a bound can be obtained as in deriving \eqref{eq:form bound below} but keeping the positive kinetic energy. Hence applying it to $ \T_\mu[\psi_j] $, we get
  		\bdm
  			(1 - \alpha) \lf\| \psi_j \ri\|^2_{H^1(\R^{dN})} \leq  \inf_{\lf\| \psi \ri\|_2 = 1} \T_{\mu}[\psi] + \tx\frac{1}{4} \delta +  C_{\alpha} \leq C,
		\edm
		for any $ \alpha > 0 $ and $ \psi_j $ normalized in $L^2$, by the boundedness from above of $ 
  		\inf_{\lf\| \psi \ri\|_2 = 1} \T_{\mu}[\psi] $.
		 
		 Now to complete the proof one needs to show that there exists another $ \bar{m} \in \N $, such that for any $ n_k \geq \bar{m} $, the opposite inequality is also true, {\it i.e.}, for any $ \delta > 0 $ arbitrary
		\beq
		 	\inf_{\lf\| \psi \ri\|_2 = 1} \T_{\mu}[\psi] < \inf_{\lf\| \psi \ri\|_2 = 1} \T_{\mu_{n_k}}[\psi] + \delta,
		\eeq 
		for $ n_k \geq \bar{m} $. The argument is, however, perfectly symmetric: pick a minimizing sequence $ \big\{ \psi_{j}^{(n_k)} \big\}_{j \in \N} $ for $ \T_{\mu_{n_k}}[\psi] $, the estimates leading to \eqref{eq:formub} can be proven in the very same way. The only one requiring some comment is the uniform boundedness of the $ H^1 $ norm of $  \psi_{j}^{(n_k)} $, which is however implied by \eqref{eq:formlb}, as above.
	\end{proof}
	
	The rest of the proof is identical to the one provided in the previous Sect. \ref{sec:gs nelson}, apart from the fact that in certain estimates the operator inequalities have to be replaced by the corresponding ones in terms of quadratic forms. We omit the details for the sake of brevity.

\appendix
\section{Fock space estimates.}
\label{sec:appendix}

In this Appendix we collect some rather standard results and technical estimates relative to the well-posedness of the models considered in the paper. In particular we provide a full proof of Propositions \ref{pro:sa1} and \ref{pro:sa2}.

The starting point is obviously the definition of the models: even the rigorous meaning of the formal expressions mentioned in Sects. \ref{sec:introduction} and \ref{sec:main} deserves some brief discussion. Indeed, even if, for the Nelson model, the operator $ H $ can be given a meaning at least as symmetric operator with dense domain, the same is not true for the polaron. We recall that the formal expression we want to study has the form
\bdm
	H = \hamf + \sum_{j = 1}^N A(g(\xv_j)),	\qquad		\hamf = \HH_0 + \diff\Gamma(\omega),
\edm
\bdm
	\HH_0 = \sum_{j = 1}^N {\lf( - \Delta_j \ri) + U(\xv_1, \ldots, \xv_N)}.
\edm
The precise assumptions made on the quantities involved are specified in Sect. \ref{sec:model} for the various models under investigation. We simply recall that $ U $ decomposes as
\bdm
	U = U_+ + U_{\ll},		\qquad		U_+ \in L^2\big(\R^{dN};\R^+\big),	\quad	U_{\ll} \in K_{\ll}\big(\R^{dN}\big).
\edm
For the field part, the only assumptions on $ \omega = \omega(\kv) $ is that it is positive.

In order to study the operator $ H $, one is typically forced to  start with the quadratic form \eqref{eq:formQ} associated to it, {\it i.e.}, 
\bdm
	Q_H[\Psi] = \meanlr{\Psi}{H}{\Psi}_{ L^2\otimes \fock} = Q_{\hfree}[\Psi] + Q_A[\Psi].
\edm
Now, under the assumptions \eqref{eq:V}, \eqref{eq:omega}, \eqref{eq:g1} or \eqref{eq:g2}, or, more importantly, under the assumptions for the polaron given by \eqref{eq:omega1} and \eqref{eq:g3}, the above expression makes sense on a dense domain given by $ C_0^{\infty}\lf(\mathbb{R}^{dN}\ri) \cap \dom(\sqrt{\mathrm{d}\Gamma(1)})\cap
 \dom(\sqrt{\mathrm{d}\Gamma(\omega)}) $. Note that on the same domain one is also allowed to split the quadratic form into two terms which make sense individually. In spite of seeming natural, this is indeed impossible at the operator level, without a closer inspection of its properties.

There are various techniques to prove that the form $ Q_H $ is associated to a unique self-adjoint operator, which will be identified with the formal expression $ H $. Depending on
the regularity of the function $ g $ appearing in the interaction term, one can apply either the Kato-Rellich Theorem (Nelson model) or the KLMN Theorem (polaron). 

Let us start by discussing the first approach, which will lead to the proof of Proposition \ref{pro:sa1}: the result relies on a technical estimate that we state in a separate Lemma \ref{lemma:est nelson}, whose discussion is postponed.

	\begin{proposition}[Kato-Rellich]
		\label{pro:kato-rellich}
		\mbox{}	\\
  		For any function $ g \in L^{\infty} \lf(\mathbb{R}^{d};\mathfrak{H} \ri) $ such that
                $\omega^{-1/2}g\in L^{\infty} \lf(\mathbb{R}^{d};\mathfrak{H} \ri)$, the operator $H$ is self-adjoint on
                $ \dom(-\Delta+{U_+})\cap \dom\bigl(\mathrm{d}\Gamma(\omega) \bigr)$ and bounded from below.
	\end{proposition}
	
	\begin{proof}
  		Using \eqref{eq:est nelson} and the identity
  		\begin{equation}
  			\label{eq:11}
  			\lf\| a^{\dagger}(g) \Psi \ri\|_{\fock(\mathfrak{H})}^2 = \lf\| a(g) \Psi \ri\|_{\fock(\mathfrak{H})}^2 + \lf\| g  \ri\|_{\mathfrak{H}}^2 \lf\| \Psi \ri\|_{\fock(\mathfrak{H})}^2,
		\end{equation} 
		we obtain for any $ \Psi \in \dom ( \sqrt{\mathrm{d}\Gamma(\omega)} )$
		\bmln{
      		\lf\| \sum A(g(\xv_j)) \Psi \ri\|_{L^2 \otimes \fock} \leq N \lf(2\lf\| \omega^{-1/2} g \ri\|_{L^{\infty} (\mathbb{R}^{d};\hh)}  \lf\| \sqrt{\mathrm{d}\Gamma(\omega)} \Psi  \ri\|_{L^2 \otimes \fock} \ri .\\+ \lf\| g \ri\|_{L^{\infty} (\mathbb{R}^{d};\hh)}\lf\| \Psi \ri\|_{L^2 \otimes \fock} \biggr).
    		}
  		Then by Cauchy-Schwarz one gets
  		\bmln{
  			2N \lf\| \omega^{-1/2} g \ri\|_{L^{\infty} (\mathbb{R}^{d};\hh)} \lf\| \sqrt{\mathrm{d}\Gamma(\omega)} \Psi  \ri\|_{L^2 \otimes \fock} \\
  			= 2N \lf\| \omega^{-1/2} g \ri\|_{L^{\infty} (\mathbb{R}^{d};\hh)} \lf( \meanlr{\Psi}{\mathrm{d}\Gamma(\omega)}{\Psi}_{L^2 \otimes \fock} \ri)^{1/2}	\\
  			\leq 2N \lf\| \omega^{-1/2} g \ri\|_{L^{\infty} (\mathbb{R}^{d};\hh)} \lf( \meanlr{\Psi}{-\Delta + U_+ + \diff\Gamma(\omega)}{\Psi}_{L^2 \otimes \fock} \ri)^{1/2} \\
  			\leq \alpha \lf\| \lf( -\Delta + U_+ + \diff\Gamma(\omega) \ri) \Psi \ri\|_{L^2 \otimes \fock} + N^2 \alpha^{-1} \lf\| \omega^{-1/2} g \ri\|_{L^{\infty} (\mathbb{R}^{d};\hh)}^2 \lf\| \Psi \ri\|_{L^2 \otimes \fock}	
   		}
 		 for any $\alpha>0$, and therefore
  		\beq
  			\lf\| \sum A(g(\xv_j)) \Psi \ri\|_{L^2 \otimes \fock} \leq \alpha \lf\| \lf( -\Delta + U_+ + \diff\Gamma(\omega) \ri) \Psi \ri\|_{L^2 \otimes \fock} + b(\alpha) \lf\| \Psi \ri\|_{L^2 \otimes \fock},
  		\eeq
  		with
  		\bdm
  			b(\alpha) =N^2 \alpha^{-1} \lf\| \omega^{-1/2} g \ri\|_{L^{\infty} (\mathbb{R}^{d};\hh)}^2 + N \lf\| g \ri\|_{L^{\infty} (\mathbb{R}^{d};\hh)}\; .
		\edm
		In addition, for any $\alpha'>0$, there exists $b'(\alpha')>0$ finite, such that
  		\begin{equation*}
    			\lf\| U_{\ll} \psi \ri\|_{L^2(\R^{dN})} \leq \alpha^{\prime} \lf\| \lf( - \Delta + U_+ \ri) \psi \ri\|_{L^2(\R^{dN})} + b'(\alpha') \lf\| \psi \ri\|_{L^2(\R^{dN})}
  		\end{equation*}
  		thanks to the hypothesis of Kato smallness of $ U_{\ll} $. The positivity of $ \diff \Gamma(\omega) $ allows to extract from the above inequality a bound on $ L^2 \otimes \fock $ where the operator on the r.h.s. is replaced with $ - \Delta + U_+ + \diff \Gamma(\omega) \geq - \Delta + U_+ $.
  		
  		Picking now $ \alpha $ and $ \alpha' $ both strictly smaller than $ \frac{1}{2} $, the result is proven by a direct application of Kato-Rellich Theorem. As a by-product we also obtain that the full operator is bounded from below by
  		\begin{equation*}
   			-M : =\sup_{1/4< \alpha,\alpha'<1/2} -\bigl(b(\alpha)+b'(\alpha')\bigr)\; > - \infty.
  		\end{equation*}
	\end{proof}
	
	The above proof does not cover the case of Fr\"{o}hlich's polaron, since in that case 
	\bdm
		g(\xv) = \frac{1}{|\kv|^{\frac{d-1}{2}}} e^{i \kv \cdot \xv},
	\edm
	which is not a function in $ L^{\infty}(\R^d;L^2(\R^d)) $. However, $ (|\kv|^2 + 1)^{-1/2} g(\xv) \in L^2(\R^d) $ for a.e. $ \xv \in \R^d $, or, in other words,
	\beq
		\label{eq:g polaron sobolev}
		\frac{1}{|\kv|^{\frac{d-1}{2}}} e^{i \kv \cdot \xv} \in W^{-1,\infty}(\R^d, L^2(\R^d)).
                \eeq
                Using the trick described in Remark \ref{rem:classical limit polaron}, \emph{i.e.}, rewriting  $ g(\xv) $ above as the commutator between $ - i \nabla_{\xv} $ and a (vector-valued) function in $ L^{\infty}(\R^d;\hh\otimes \R^d) $, it is possible to exploit the aforementioned regularity \eqref{eq:g polaron sobolev} and prove self-adjointness of $ H$ via the KLMN Theorem.
                
                Before stating the result, however, we remark that the function \eqref{eq:g polaron sobolev} not only is in $ W^{-1,\infty}(\R^d, \hh) $ but possesses a stronger regularity property: for any $ \delta > 0 $, there exists a $ \varrho > 0 $, such that one can decompose
                \beq
		\label{eq:g decomposition 1} 
		g(\xv) = g_{<, \varrho}(\xv) + g_{>,\varrho}(\xv), 	\qquad		\omega^{-1/2} g_{<,\varrho} \in 	L^{\infty}(\R^d,\hh),	\quad	\omega^{-1/2} g_{>,\varrho} \in W^{-1,\infty}(\R^d,\hh),
	\eeq	
	with
	\beq
		\label{eq:g decomposition 2}
		\lf\|\omega^{-1/2}\lf( - \Delta \ri)^{-1/2} g_{>,\varrho}(\xv) \ri\|_{L^{\infty}(\R^d,\hh)} \leq \delta,
	\eeq
	{\it i.e.}, the function can be decomposed into a part which is in $ \hh $ for a.e. $ \xv $ and a rest whose (homogeneous) $ W^{-1,\infty}(\R^d, \hh) $ norm can be assumed to be arbitrarily small. Concretely, in the case of the polaron, this can be easily realized by writing
	\bdm
		g(\xv) = \one_{[0,\varrho]}(|\kv|) g(\xv) + \one_{[\varrho, + \infty)}(|\kv|) g(\xv) = : g_{<, \varrho}(\xv) + g_{>,\varrho}(\xv),
	\edm
	and taking $ \varrho $ large enough. Indeed, since $\omega=1$ for the polaron, one has
	\bdm
		\lf\| \omega^{-1/2}\lf|\kv\ri|^{-1} g_{>,\varrho}(\xv) \ri\|_{L^{\infty}(\R^d,L^2(\R^d))}^2 = \int_{|\kv| \geq \varrho} \diff \kv \: \frac{1}{|\kv|^{d+1}} = C \varrho^{-1} \xrightarrow[\varrho \to \infty]{} 0.
	\edm

	\begin{proposition}[KLMN]
		\label{pro:klmn}
		\mbox{}	\\
  		Let $ g $ be a function such that \eqref{eq:g decomposition 1} and \eqref{eq:g decomposition 2} are satisfied. Then $ Q_H $ is closed on $ \dom(\sqrt{-\Delta+U_+})\cap \dom(\sqrt{\mathrm{d}\Gamma(\omega)})$ and bounded from below. Therefore it is associated to a unique operator $H$ self-adjoint on
  $ \dom(H) \subset \dom(\sqrt{-\Delta+U_+})\cap \dom(\sqrt{\mathrm{d}\Gamma(\omega)})$ that is also bounded from below.
	\end{proposition}

	\begin{proof}
		By linearity of $ A(g) $ we can split the quadratic form into three pieces: $ Q_{H,<} $, $ Q_{H,>} $, and $Q_{H, \ll}$ with obvious meaning of the notation. Let us consider $ Q_{H,>} $ first. An easy computation yields
  		\begin{equation*}
   			 a(g_{>, \varrho}(\xv)) = \lf[ a\lf( \nabla_\xv  \tilde{g}_{>, \varrho}(\xv)   \ri), - \nabla_\xv \ri],
  		\end{equation*}
		with $ \tilde{g}_{>,\varrho}(\xv) : = \lf( - \Delta_{\xv} \ri)^{-1} g_{>,\varrho}(\xv) $. Therefore it follows that
  		\bmln{
      		\lf| Q_{H,>}[\Psi] \ri| \leq 2 \sum_{j=1}^N \lf| \braketr{- i \nabla_{\xv_j} \Psi}{ a\lf( \nabla_{\xv_j}  \tilde{g}_{>, \varrho}(\xv_j)   \ri) \Psi}_{L^2 \otimes \fock} \ri| \\
      		\leq 4 N \lf\|\omega^{-1/2}\lf( - \Delta \ri)^{-1/2} g_{>,\varrho}(\xv) \ri\|_{L^{\infty}(\R^d,\hh)}\lf\| \sqrt{\mathrm{d}\Gamma(\omega)} \Psi  \ri\| \lf\| \sqrt{- \Delta + U_+} \Psi \ri\| \\
      	\leq 4N \lf\|\omega^{-1/2}\lf( - \Delta \ri)^{-1/2} g_{>,\varrho}(\xv) \ri\|_{L^{\infty}(\R^d,\hh)} Q_{- \Delta + U_+ + \diff \Gamma(\omega)}[\Psi],
    		}
  		for any $ \Psi \in \dom(\sqrt{-\Delta+V_1})\cap \dom(\sqrt{\mathrm{d}\Gamma(\omega)})$, where we have used Cauchy-Schwarz inequality and \eqref{eq:est nelson}. Now by \eqref{eq:g decomposition 2}, it is possible to choose $ \varrho > 0$ big enough such that
  		\begin{equation}
                  \label{eq: rho}
    			\alpha_1(\varrho) : = 4N \lf\| \omega^{-1/2}\lf( - \Delta \ri)^{-1/2} g_{>,\varrho}(\xv) \ri\|_{L^{\infty}(\R^d,\hh)}= \tx\frac{1}{4}.
  		\end{equation}
  		Now let us turn the attention to $ Q_{H,<} $, with $\varrho$ fixed by condition \eqref{eq: rho}. Using again the Cauchy-Schwarz inequality and \eqref{eq:est nelson}, we obtain
  		\bmln{
  			\lf| Q_{H,<}[\Psi] \ri| \leq \alpha
     \meanlr{\Psi}{\diff \Gamma(\omega)}{\Psi}_{L^2 \otimes \fock} + N^{2} \alpha^{-1}\lf\| \omega^{-1/2} g_{<,\varrho}(\xv) \ri\|_{L^{\infty}(\R^d,\hh)} \lf\| \Psi \ri\|_{L^2\otimes\fock}^2 \\
      		\leq \alpha\, Q_{- \Delta + U_+ + \diff \Gamma(\omega)}[\Psi]+ N^{2} \alpha^{-1}\lf\| \omega^{-1/2} g_{<,\varrho}(\xv) \ri\|_{L^{\infty}(\R^d,\hh)} \lf\| \Psi \ri\|_{L^2\otimes\fock}^2
   		}
  		for any $\psi\in \dom(\sqrt{-\Delta+U_+})\cap \dom(\sqrt{\mathrm{d}\Gamma(\omega)})$ and $\alpha>0$. Choosing $\alpha=  \frac{1}{4}$ we obtain
  		\bml{
    			\label{eq:klmn bound}
      		\lf| Q_{H,>}[\Psi]+ Q_{H,<}[\Psi] \ri|  \leq \tfrac{1}{2}\, Q_{-\Delta + U_+ + \diff \Gamma(\omega)}[\Psi] \\
      		+ 3N^{2} \lf\| \omega^{-1/2} g_{<,\varrho}(\xv) \ri\|_{L^{\infty}(\R^d,\hh)}  \lf\| \Psi \ri\|_{L^2\otimes\fock}^2,
    		}
    		for a suitable $ \varrho > 0 $ such that \eqref{eq: rho} is satisfied. Now since $ U_{\ll} $ is infinitesimally form-bounded w.r.t. $-\Delta$, it follows that there exists a constant $C_{\ll}>0$ such that for any $ \Psi \in \dom(\sqrt{-\Delta+V_1})\cap \dom(\sqrt{\mathrm{d}\Gamma(\omega)})$,
                \begin{equation*}
                  \lf| Q_{H, \ll}[\Psi]\ri|\leq \tfrac{1}{3}Q_{-\Delta + U_+ + \diff \Gamma(\omega)}[\Psi] + C_{\ll}\lf\| \Psi \ri\|_{L^2\otimes\fock}^2\; .
                \end{equation*}
                The result then follows from KLMN Theorem.
	\end{proof}
	
	We conclude the Appendix with a technical estimate used before in the proofs of both Propositions \ref{pro:kato-rellich} and \ref{pro:klmn}:

	\begin{lemma}
		\label{lemma:est nelson}
		\mbox{}	\\
		For any  function $g$ such that $\omega^{-1/2} g \in L^{\infty} \lf(\mathbb{R}^{d};\mathfrak{H} \ri) $, we have
		\begin{equation}
  			\label{eq:est nelson}
  			\lf\| a(g(\xv)) \Psi \ri\|_{L^2(\R^{dN}) \otimes \fock(\hh)} \leq \lf\| \omega^{-1/2} g \ri\|_{L^{\infty}(\R^d; \hh)} \lf\| \sqrt{\diff \Gamma(\omega)} \Psi \ri\|_{L^2(\R^{dN}) \otimes \fock(\hh)}.
		\end{equation}
	\end{lemma}
	
	\begin{proof}
		The result is obtained via Cauchy inequality as in the derivation of \eqref{eq:fock est form}. We omit the details.	
	\end{proof}

\newcommand{\etalchar}[1]{$^{#1}$}
\providecommand{\bysame}{\leavevmode\hbox to3em{\hrulefill}\thinspace}
\providecommand{\MR}{\relax\ifhmode\unskip\space\fi MR }
\providecommand{\MRhref}[2]{%
  \href{http://www.ams.org/mathscinet-getitem?mr=#1}{#2}
}
\providecommand{\href}[2]{#2}

\end{document}